\documentclass[11pt]{article}

\usepackage[letterpaper, margin=1in]{geometry}
\usepackage{booktabs}

\usepackage{authblk}

\usepackage{todonotes}
\usepackage{microtype}
\usepackage{xspace}
\usepackage{comment}
\usepackage{amsmath,amsthm,amssymb}
\usepackage{enumerate}
\usepackage{verbatim}
\usepackage{hyperref}
\usepackage{amsmath,amsthm,amssymb}
\usepackage{xspace}
\usepackage{url}
\usepackage{enumerate}
\usepackage{hyperref}
\usepackage{verbatim}
\usepackage{tikz}
\usepackage{array}

\allowdisplaybreaks

\usepackage{color}

\definecolor{darkgreen}{rgb}{0.0,0.7,0.0}
\newenvironment{aj}{\noindent\color{magenta} Artur:} {}

\newenvironment{mg}{\noindent\color{blue} Micha?:} {}

\newcommand{\procfont}[1]{\textsc{\textrm{#1}}}
\newcommand{\hk}[2][k]{{|#2|H_{#1}(#2)}}

\newcommand{\seq}[3]{\ensuremath{{#1}_{#2},\ldots,{#1}_{#3}}}

\newcommand{\en}[2][0]{\ensuremath{|#2|H_{#1}(#2)}}
\newcommand{\T}{{\mathcal T}}

\newcommand{\algofont}[1]{\textnormal{\selectfont\sffamily#1}}
\providecommand{\twodots}{\mathinner{\ldotp\ldotp}}

\providecommand{\Ocomp}{{\mathcal O}}

\newcommand{\prob}{\mathbb{P}}

\newcommand{\Cref}[1]{(\hyperref[C#1]{C#1})\xspace}
\newcommand{\Crefall}{(\hyperref[C1]{C1--C4})\xspace}

\newcommand{\Tref}[1]{(\hyperref[T#1]{T#1})\xspace}
\newcommand{\Trefall}{(\hyperref[T1]{T1--T4})\xspace}

\theoremstyle{plain}
\newtheorem{lemma}{Lemma}
\newtheorem{theorem}[lemma]{Theorem}

\theoremstyle{definition}
\newtheorem{definition}[lemma]{Definition}
\theoremstyle{remark}
\newtheorem{example}[lemma]{Example}

\newcommand{\ts}{|\mathcal{T}|}
\newcommand{\tps}{|\mathcal{T}'|}

\usepackage{titlesec}

\titlespacing{\paragraph}{0pc}{\parskip}{1pc}

\begin{document}
	

\title{Using statistical encoding to achieve tree succinctness never seen before}

\author{Micha{\l} Ga{\'n}czorz}
\affil{University of Wroc{\l}aw, Wroc{\l}aw, Poland, mga@cs.uni.wroc.pl}


\maketitle
\begin{abstract}
We propose a new succinct representation of labeled trees
which represents a tree $\mathcal T$ using $\hk{\mathcal T}$ number of bits
(plus some smaller order terms),
where $\hk{\mathcal T}$ denotes the $k$-th order (tree label) entropy,
as defined by Ferragina at al.\ 2005.
Our representation employs a new, simple method of partitioning the tree,
which preserves both tree shape and node degrees.
Previously, the only representation that used $\hk{\mathcal T}$
bits was based on \algofont{XBWT}, a~transformation that linearizes a labeled tree into a string,
combined with compression boosting.
The proposed representation is much simpler than the one based on \algofont{XBWT},
which used additional linear space (bounded by $0.01n$) hidden in the ``smaller order terms'' notion, as an artifact of using zeroth order entropy coder; 
our representation uses sublinear additional space (for reasonable values of $k$ and size of the label alphabet $\sigma$).
The proposed representation can be naturally extended to a~succinct data structure for trees,
which uses $\hk{\mathcal T}$ plus 
additional $\Ocomp(|\mathcal{T}| k\log \sigma / \log_\sigma |\mathcal{T}| 
+ |\mathcal{T}| \log \log_\sigma |\mathcal{T}| / \log_\sigma |\mathcal{T}|)$ bits
and supports all the usual navigational queries in constant time.
At the cost of increasing the query time to
$\Ocomp(\log \ts / \log \log \ts)$
we can further reduce the space redundancy to
$\Ocomp(|\mathcal{T}| \log \log |\mathcal{T}| / \log_\sigma |\mathcal{T}|)$ bits,
assuming $k \leq \log_\sigma |\mathcal{T}|$.
This is a major improvement over representation based on \algofont{XBWT}:
even though \algofont{XBWT}-based representation uses $\hk{\mathcal T}$ bits,
the space needed for structure supporting navigational queries is much larger:
original solution consumed at least $\hk[0]{\mathcal T}$ bits, i.e.\ zeroth order entropy of labels.
This was later improved to achieve $k$-th order entropy of 
string obtained from \algofont{XBWT},
see Ferragina et al. 2006,
who argued that such an entropy is intuitively connected to $\hk{\mathcal T}$, 
though they are formally different;
still, this representation gave non-constant query time and
did not support more complex queries like level\_ancestor.
Lastly, our data structure is fairly simple,
both conceptually and implementationally, and uses known tools,
which is a counter-argument to the claim
that methods based on tree-partitioning are impractical.


We also introduce new, finer, tree entropy measures,
which lower bound previously known ones.
They try to capture both the shape of the tree, which previously was measured by tree entropy (Jansson et al. 2006),
and $\hk{\mathcal{T}}$ (Ferragina et al. 2005)
and the relation between them.
Our motivation for introducing this measure is that while $k$-th order entropy for strings captures well the compressibility
(and the compression algorithms based on them in practice outperform other compression methods in terms of compression ratio),
this is not true in case of trees:
simple dictionary-based methods for trees,
like tree grammars or Top-Dag, work good in practice and can sometimes outperform \algofont{XBWT},
which encodes the labels using the $k$-th order label entropy
(and separately encodes the tree shape).
Considering labels and tree shape together should give a better notion.
All bounds on the space usage of our structures hold also using the new measures,
which gives hope for a better practical performance.
\end{abstract}

\newpage
\section{Introduction}
Labeled trees arise in many applications in the real world,
with XML files being the most known one.
A trivial encoding of $n$ labels from a $\sigma$-size alphabet uses 
$n \log \sigma$ bits;
naturally, we would like to encode them more succinctly,
without loosing the possibility of efficiently querying those trees.
%
The first nontrivial analysis of encodings of labels was 
carried out by Ferragina at al.~\cite{ferragina2005labeledTrees},
who introduced the concept of $k$-th order entropy of trees,
denoted by $H_k$, and proposed a novel transformation, \algofont{XBWT},
which can be used to represent the tree using $\hk[k]{\T}$
plus smaller order terms bits.
The $H_k$ is a natural extension of $k$-th order entropy for words,
i.e.\ it is defined in similar manner,
where the context is a concatenation of $k$ labels
on the path to the root.
Moreover, it is claimed to be good measure both in
practice~\cite{ferragina2009compressing}
(as similarly labeled nodes
descend from similarly labeled contexts~\cite{ferragina2005labeledTrees})
and in theory, as it is related to theory of tree sources~\cite{SourceEncodingSyntactic}.
\algofont{XBWT} generates a single string from tree labels
by employing a \algofont{BWT}-like transformation on trees
(and it stores the tree shape separately).
To achieve $H_k$, the compression boosting technique,
known from text compression~\cite{FerraginaCompressionBoosting},
is also needed.
This has some disadvantages:
compression boosting adds a small linear term (bounded by $0.01n$) to the redundancy~\cite{ferragina2005labeledTrees},
which causes this solution to be a~little off from optimal
Moreover, this approach works for tree compression,
to perform operations on such representation additional structures
for \procfont{rank/select} over alphabet of size $\sigma$
are needed.
Later on Ferragina et al.~\cite{ferragina2009compressing}
proposed a solution which
encodes the labels of the tree using $k$-th order entropy of string
produced by \algofont{XBWT} and support some navigational queries,
this result may seem unrelated,
but they argued that intuitively
$k$-th order entropy of the tree and entropy of \algofont{XBWT}
string are similar, because they both similarly cluster node labels
(though no formal relation is known between the two).
Still, this structure supports very limited set of operations,
for instance, it does not support
\procfont{level\_ancestor} queries
has non-constant query time and large additional space consumption.
The latter is due to the fact that it needs high-order compressed 
structure for \procfont{rank/select} queries over large alphabets,
which was proven to be much harder than simple \procfont{rank/select}~\cite{OptimalLowerUpper, GrossiOptimalTradeoffIndexes}.
Note though, that this structure supports also
basic label-related operations like \procfont{childrank($v, \alpha$)},
which returns rank of $v$ among children labeled with $\alpha$.
We note that, until this paper, \algofont{XBWT} was the only known 
compression method achieving $H_k$, for $k > 0$, for trees.

All of the above mentioned methods still needs additional $2\ts$ bits
for storing tree shape.
Later on, it was observed that we can beat $2\ts$ bits
bound for storing tree shape, by assuming some
distribution on node degrees~\cite{ultraSuccintTrees}.
This was achieved by encoding tree in DFUDS representation.
The formal definition of tree entropy
provided by Jansson et al.~\cite{ultraSuccintTrees}
is simply the classical zeroth order
entropy on the node degrees.

Aside from entropy-based tree compression,
dictionary tree compression methods were also proposed;
those methods try to exploit similarity of different subtrees.
This includes top-trees~\cite{bille2013TopTrees}
and tree grammars~\cite{BuLoMa07}.
Interestingly, while context-based methods on strings, like
\algofont{PPM}/\algofont{BWT}, are in practice clearly superior
to dictionary-based ones
(and also in theory, as dictionary methods can give significantly worse bit-size with respect to $H_k$,
under the same assumptions~\cite{entropyBoundsUnpublished}),
we cannot claim the same for trees~\cite{ferragina2009compressing, LohreyMM11},
i.e\ methods achieving $H_k$ for trees are not significantly better in practice.
We believe one reason behind this phenomenon is that current entropy measures for trees
do not capture all of the information stored in a tree:
for instance \algofont{XBWT} based methods, which are optimal with respect to $H_k$,
ignore the information on tree shape,
whereas in real life data tree labels and tree shape are highly correlated.

\paragraph{Our contribution}

We start by defining two new measures of entropy of trees,
which take into the account \emph{both} tree structure \emph{and} tree labels:
$H_k(\mathcal{T}|L)$ and $H_k(L|\mathcal{T})$.
Those measures lower bound
previous measures, i.e
tree entropy~\cite{ultraSuccintTrees}
and $H_k$~\cite{ferragina2005labeledTrees}, respectively.

Then we propose a new way of partitioning the tree.
In contrast to previous approaches,
(i.e.\ succinct representations~\cite{geary2006succinct,
farzan2009universal}
and dictionary compression),
this partition preserves \emph{both} the shape of the tree
\emph{and} the degrees of the nodes.
We show that by applying entropy
coder to tree partition
we can bound the size of the tree encoding by both
$\ts H_k(\mathcal{T}|L) + \ts H_k(L) +
\Ocomp(\ts k\log\sigma / \log_\sigma \ts + \ts \log\log_\sigma \ts / \log_\sigma \ts)$
and $\ts H_k(L|\mathcal{T})  + \ts H(\mathcal{T}) + 
\Ocomp(\ts(k+1)\log\sigma / \log_\sigma \ts + \ts\log\log_\sigma \ts / \log_\sigma \ts)$
bits.
This is the first method which is not based on \algofont{XBWT}
and achieves bounds related to $H_k$.

We show that using standard techniques
we can augment our tree encoding, at the cost of increasing the constants
hidden in the $\Ocomp$ notation,
so that most of the navigational queries are supported in constant time,
thus getting the first structure which achieves $H_k$ for trees
and supports queries on compressed form.
Note that some of the previously mentioned methods, like tree grammars~\cite{ganardi2017constructing} or
high-order compressed \algofont{XBWT}~\cite{ferragina2005labeledTrees},
do not have this property.

Then we show that we can further reduce the redundancy to $\Ocomp(\ts \log \log \ts / \log_\sigma \ts)$
bits, at the cost of increasing query time to $\Ocomp(\log \ts / \log \log \ts)$.
Moreover, we show that we use known tools to support label-related operations,
we get the same additional cost as previously mentioned structure based
on \algofont{XBWT}~\cite{ferragina2009compressing},
but still supporting the rest of the operations in constant time,
thus we get the structure which outperforms previously known ones.
Our methods can also be applied to unlabeled trees to achieve tree entropy~\cite{ultraSuccintTrees},
in which case we get the same (best known) additional space as in~\cite{ultraSuccintTrees},
our techniques, in case of standard operations, are also less complex
than other methods achieving tree entropy~\cite{ultraSuccintTrees, farzan2014uniform}.
Lastly, our structure allows to retrieve the tree
in optimal, $\Ocomp(\ts / \log_\sigma \ts)$ time
(assuming machine words of size $\Theta(\log \ts)$).

\section{Definitions}\label{sec:definitions}
We denote the input alphabet by $\Sigma$,
and size of the input alphabet as $\sigma = |\Sigma|$.
For a tree $\mathcal{T}$ we denote by $|\mathcal{T}|$
the number of its nodes and the same applies to forests.
We consider rooted (i.e.\ there is a designated root vertex),
ordered (i.e.\ children of a given vertex have a left-to-right order imposed on them)
$\Sigma$-labeled (i.e.\ each node has a label from $\Sigma$) trees,
moreover label does not determine node degree and vice versa.
We assume that bit sequences of length $\log \ts $ fit into $\Ocomp(1)$ machine words
and that we can perform operations on them in constant time.

\subsection{String Entropy and k-th order Entropy}
For a string $S$ its (\emph{zeroth order}) \emph{entropy}, denoted by $H_0$,
is defined as $\en{S} =  -\sum_{s \in S} t_s \log \frac{t_s}{|S|}$,
where $t_s$ is a number of occurrences of character $s$ in $S$.
It is convenient to think that $- \log \frac{t_s}{|S|}$ assigned to a symbol $s$
is the optimal cost of encoding this symbol (in bits)
and $\frac{t_s}{|S|}$ is the \emph{empirical probability} of occurrence of $s$.
Note that those ``values'' are usually not natural numbers.

The standard extension of this measure is the $k$-th order entropy,
denoted by $H_k$,
in which the (empirical) probability of $s$ is conditioned by $k$ preceding letters,
i.e.\ the cost of single occurrence of letter $s$ is equal to $\log \prob(s|w)$,
where $\prob(s|w) = \frac{t_{ws}}{t_w}$, $|w|=k$ and 
$t_{v}$ is the number of occurrences of a word $v$ in given word $S$.
We call $\prob(s|w)$ the \emph{empirical probability} of a letter $s$ occurring 
in a $k$-letter context.
Then
$\en[k]{S} = -\sum_{s \in \Sigma, w \in \Sigma^k } t_{ws} \log \prob(s|w)$.
The cost of encoding the first $k$ letters is ignored
when calculating the $k$-th order entropy.
This is acceptable, as $k$ is (very) small compared to $|S|$,
for example most tools based on popular context-based compressor family PPM
use $k\leq 16$.

\subsection{Entropy and $k$-th order entropy for trees.}

In the case of labeled trees, $0$-th order entropy
has a natural definition --- 
its a zeroth order entropy, $H_0(S)$, of string $S$ made by concatenation of labels of vertices.
However, for $k$-th order entropy the situation gets more involved,
as now we have to somehow define the context.

\paragraph{Label-entropy}
Ferragina et al.~\cite{ferragina2005labeledTrees} proposed
a definition of $k$-th order entropy of labeled trees.
The context is defined as the $k$ labels from the node to the root.
Similarly, as in the case of first $k$ letters in strings,
this is undefined for nodes whose path to the root is of length less than $k$,
which can be large, even when $k$ is small.
%
There are two ways of dealing with this problem:
in the first we allow the node to have the whole
path to the root as its context (when this path shorter than $k$);
in the second we pad the too short context with some fixed letters.
Our algorithms can be applied to both approaches,
resulting in the same (asymptotic) redundancy;
for the sake of the argument we choose the first one,
as the latter can be easily reduced to the former.

Observe that such defined entropy takes into the account
only the labels, and so we call it the 
\emph{$k$-th order entropy of labels},
which is formally defined as
$|L|H_k(L) = -\sum_{v \in \mathcal{T} } \log \prob(l_v | K_v)$,
where $l_v$ is label of vertex $v$,  $K_v$ is the word made by 
last $k$ labels of nodes on the path from root of $\mathcal{T}$ to $v$
(or less if the path from the root to $v$ is shorter than $k$)
and as in the case for string $\prob(l_v | K_v)$ is the empirical
probability of label $l_v$ conditioned that it occurs in context $K_v$.

\paragraph{Tree entropy}
The ($k$-th order) entropy of labels ignores the shape of the tree
and the information carried by it,
and we still need to represent the tree structure somehow.
A~counting argument shows that
representing an unlabeled tree of $n$ vertices
requires $2n - \Theta(\log n)$ bits~\cite{JacobsonSuccint}
and there are practical ways of storing unlabeled trees
using $2n$ bits (balanced parenthesis~\cite{jacobson1989space},
DFUDS~\cite{benoit2005representing},
LOUDS~\cite{jacobson1989space}),
which are often employed in succinct data structures for trees~\cite{succintDictionariesWithTrees,
fullyFunctionalTrees, farzan2009universal, munro1997succinct, raman2003succinct,
benoit2005representing, jacobson1989space}.

Yet, as in the case of strings, real data is rarely a random tree drawn
from the set of all trees:
For example XML files are shallow and some tree
shapes repeat, like in Figure~\ref{fig:catalog}.
For this reason the notion of \emph{tree entropy} was introduced~\cite{ultraSuccintTrees},
with the idea that it takes into the account
the probability of a node having a particular degree,
i.e.\ it measures the number of trees under some degree distribution.
Formally, it is defined as: $|\mathcal {T}| H(\mathcal{T}) = -\sum_{i=0}^{|\mathcal T|} d_i \log \frac{d_i}{|\T|}$,
where $d_i$ is the number of vertices of degree $i$ in $\mathcal T$.
Up to $\Theta (\log n)$ additive summand,
tree entropy is an information-theoretic lower bound
on the number of bits needed to represent unlabeled tree that has
some fixed degree distribution~\cite{ultraSuccintTrees}.
Moreover, as in the case of string entropy,
tree entropy lower bound the simpler estimations,
in the sense that $|\T| H(\T) \leq 2 \ts $~\cite{ultraSuccintTrees}.

\paragraph{Mixed entropy}
Label entropy and tree entropy treat labels and tree structure separately,
and so did most of the previous approaches to labeled tree data
structures~\cite{ferragina2005labeledTrees, ultraSuccintTrees,
ferragina2009compressing, he2014framework}.
Yet, the two are most likely correlated:
one can think of XML document representing the collection of
different entities such as books, magazines etc.,
see Figure~\ref{fig:catalog}.
Knowing that the label of some node is equal to ``book'',
and that each book has an author, a year and a title,
determines degree of the vertex to be three
(and in tree grammar compression model we explicitly assume that the label uniquely defines the arity of node~\cite{BuLoMa07}).
On the other hand, knowing the degree of the vertex
can be beneficial for information on labels.

Motivated by these considerations, we define two types of \emph{mixed entropies}:
\begin{itemize}\setlength\itemsep{-0.1em}
	\item $\ts H_k(L|\mathcal{T}) = -\sum_{v \in \mathcal{T}}
	\log \prob (l_{v} | K_{v}, d_v)$,
	where $\prob (l_{v} | K_{v}, d_v)$ is the empirical probability of
	node $v$ having label $l_v$ conditioned that it occurs in the context $K_v$ and
	the node degree is $d_v$,
	that is $\prob (l_{v} | K_{v}, d_v) =
	\frac{t_{K, l_v, d_v}}{ t_{K, d_v}}$,
	where $t_{K, l_v, d_v}$ is a number of nodes in $\mathcal{T}$
	with context $K$,
	having degree $d_v$ and a label $l_v$,
	and $t_{K, d_v}$ is number of nodes in $\mathcal{T}$
	preceded by the context $K$,
	and having degree $d_v$;
	\item $\ts H_k(\mathcal{T}|L) = -\sum_{v \in \mathcal{T}}
	\log \prob (d_v | K_{v}, l_v)$,
	where $\prob (d_v | K_{v}, l_v)$,
	is empirical probability of node $v$ having
	degree $d_v$ conditioned that $v$ has a context $K_v$ and a label~$l_v$.
	The formal definition is similar to the one above.
\end{itemize}

We show that we can represent a tree using
either $\ts  H_k(\mathcal{T}|L) + \ts H_k(L)$ or $\ts H_k(L|\mathcal{T}) + \ts H(\mathcal{T})$ bits
(plus some small order terms).

On the other hand, the new measures
lower bound old ones, i.e.\ tree entropy and label entropy:
\begin{lemma}
\label{lem:entropies_not_larger}
The following inequalities hold:
$$
\ts H_k(\mathcal{T}|L) \leq \ts H(\mathcal{T})
\quad \text{and} \quad
\ts H_k(L|\mathcal{T}) \leq \ts H_k(L) \enspace .
$$
\end{lemma}
Lemma~\ref{lem:entropies_not_larger} follows from a standard use of \emph{log sum inequality}:
intuitively increasing number of possible contexts can only reduce entropy,
see Section~\ref{sec:definitions_appendix}.

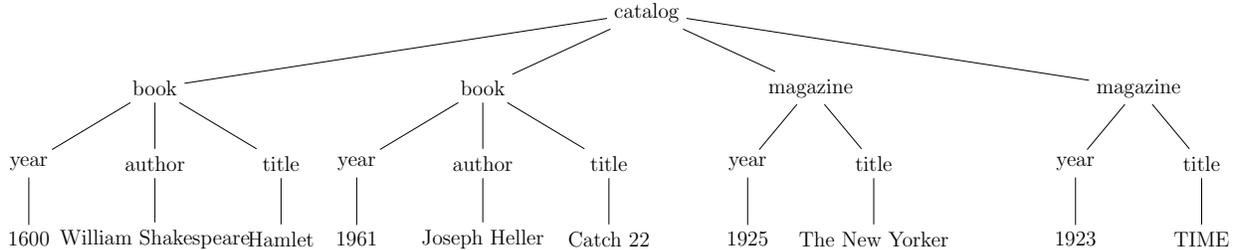
\begin{figure*}
\begin{center}
\resizebox{1.0\textwidth}{!}{%
\begin{tikzpicture}[
level1/.style ={sibling distance=6.5cm},
level2/.style ={sibling distance=2.5cm},
level3/.style ={sibling distance=1.5cm},
]
\large
\node {catalog}
	child[level1] { node {book}
		child[level2] { node {year} child[level3] { node {1600} } }
		child[level2] { node {author} child[level3] { node {William Shakespeare} } }
		child[level2] { node {title} child[level3] { node {Hamlet} } }
	}
	child[level1] { node {book}
		child[level2] { node {year} child[level3] { node {1961} } }
		child[level2] { node {author} child[level3] { node {Joseph Heller} } }
		child[level2] { node {title} child[level3] { node {Catch 22} } }
	}
	child[level1] { node {magazine}
		child[level2] { node {year} child[level3] { node {1925} } }
		child[level2] { node {title} child[level3] { node {The New Yorker} } }
	}
	child[level1] { node {magazine}
		child[level2] { node {year} child[level3] { node {1923} } }
		child[level2] { node {title} child[level3] { node {TIME} } }
	}
	;
\end{tikzpicture}
}
\end{center}
\vspace{-0.5cm}
\caption{Sample XML file structure}
\label{fig:catalog}
\end{figure*}

\section{Tree clustering}\label{sec:tree_clustsering}
We present new clustering method which preserves both node labels and vertex degrees.
\paragraph{Clustering}
The idea of our clustering technique is that we group nodes
into clusters of $\Theta(\log_\sigma \ts)$ nodes,
and collapse each cluster into a single node,
thus obtaining a tree $\mathcal{T'}$ of $\Ocomp(\ts / \log_\sigma \ts)$ nodes.
We label its nodes so that the new  label uniquely determines
the cluster that it represents
and separately store the description of the clusters.

The idea of grouping nodes was used before in the context of compressed tree indices~\cite{geary2006succinct, he2007succinct},
also some dictionary compression methods like tree-grammars or top-trees and other carry some similarities~\cite{bille2013TopTrees, ganardi2017constructing,
hubschle2015tree,treeLZ}.
Yet, from our perspective, their main disadvantage is that the degrees of the nodes in the internal representation
were very loosely connected to the degrees of the nodes in the input tree, thus the tree entropy and mixed entropies are hardly usable in upper bounds on space usage.
We propose a new clustering method, which preserves the degrees of the nodes and tree structure much better:
most vertices inside clusters have the same degree as in $\mathcal{T}$ and those that do not have their degree zeroed.
Thus, when the entropy coder is used to represent the clusters
we can bound the used space both in terms
of the label/tree entropy and mixed entropy of $\mathcal{T}$.

The clustering uses a parameter $m$,
which is the maximum size of the cluster
(up to the factor of $2$).
Each node of the tree is in exactly one cluster
and there are two types of nodes in a cluster:
\emph{port} and \emph{regular} nodes.
A port is a leaf in a cluster
and for a regular node all its children in $\T$
are also in the same cluster (in the same order as in $\T$);
in particular, its degree in the cluster is the same as in $\T$.
Observe that this implies that each node
with degree larger than $2m$
will be a port node.

The desired properties of the clustering are:
\begin{enumerate}[(C1)]\setlength\itemsep{-0.2em}
	\item\label{C1} there at least $\frac{\ts}{2m} - 1$ and at most $\frac{2\ts}{m} + 1$ clusters;
	\item\label{C2} each cluster is of size at most $2m-1$;
	\item\label{C3} each cluster is a forest of subtrees (i.e.\ connected subgraphs) of $\T$,
	roots of trees in this forest are consecutive siblings in $\mathcal{T}$;
	\item\label{C4}		
	each node in a cluster $C$ is either
	a port node or a regular node;
	each port node is a leaf in $C$
	and each regular node has the same degree in $C$
	as in $\mathcal{T}$.
	In particular, if node $u$ belongs to some cluster $C$,
	then either all of its children are in $C$ or none.
\end{enumerate}
A clustering satisfying~\Crefall can be found using a natural bottom-up greedy algorithm.
\begin{lemma}
	\label{lem:clustering_procedure}
	Let $\mathcal{T}$ be a labeled tree.
	For any $m \leq \ts$ we can construct
	in linear time a~partition of nodes of $\mathcal{T}$ into clusters
	satisfying conditions~\Crefall.
\end{lemma}
\begin{figure*}[ht]
	\vspace{-0.5cm}
	\centering
	\def\svgwidth{\textwidth}
	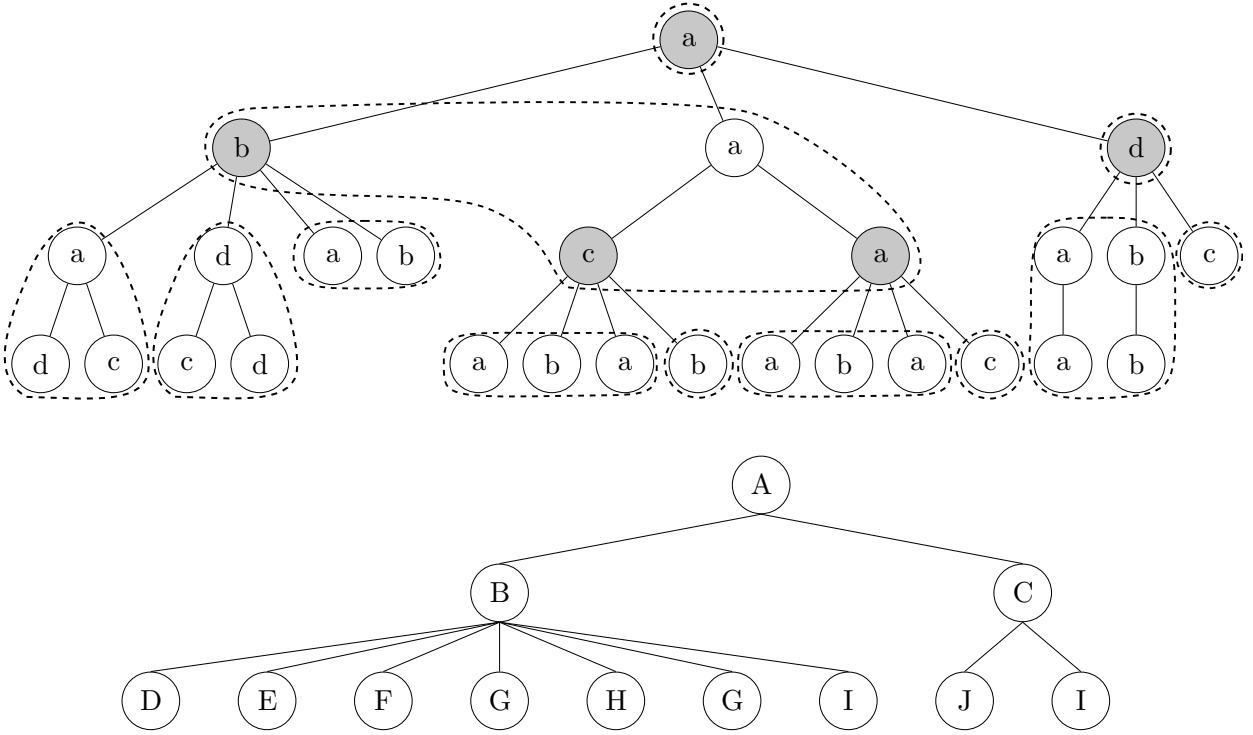
	\caption{Clustering of tree for parameter $m=3$ and tree created by replacing clusters with new nodes.
	Marked nodes are port nodes.}
	\label{fig:trees}
\end{figure*}

\paragraph{Building the cluster tree}

We build the \emph{cluster tree} out of the clustering satisfying \Crefall: we replace each cluster with a new node and put
edges between new nodes if there was an edge between some
nodes in the corresponding clusters.
To retrieve the original tree $\T$ from the cluster tree and its labels
we need to know the degree of each port node in the original
tree (note that this depends not only on the cluster, but also on the particular cluster node,
e.g. two clusters may have the same structure but can have different port node degrees in $\T$).
Thus we store for each cluster node with $k$ ports a
\emph{degree sequence} $d_1, d_2, \ldots, d_k$ where $d_i$ is an 
outdegree of $i$-th port node in the input tree (in natural left-to-right order on leaves). 
For an illustration, see cluster replaced by cluster node labeled $B$
in Figure~\ref{fig:trees}, its degree sequence is $3, 2, 2$.
In section~\ref{sec:application_compressing}
we show that degree sequence can be stored efficiently
and along with cluster tree it is sufficient to retrieve and navigate $\T$.

\begin{definition}[Cluster structure]\label{def:cluster_structure}
For a labeled tree $\mathcal{T}$ and parameter $m$ we define \emph{cluster structure},
denoted $C(\mathcal{T})$.
The cluster structure consists of:
\begin{itemize}\setlength\itemsep{-0.2em}
	\item Ordered, rooted, labeled tree $\mathcal{T'}$ (called cluster tree) with 
	$\Theta\left(\frac{\ts}{m}\right)$ nodes,
	where each node represents a cluster,
	different labels correspond to different clusters
	and the induced clustering of $\T$ satisfies~\Crefall.
	
	\item For each node $v' \in \mathcal{T'}$,
	the degree sequence $d_{{v'},1}, \ldots d_{{v'},j}$
	where $d_{{v'},i}$ indicates that $i$-th port node
	in left-to-right order on leafs of cluster
	represented by $v'$ connects to $d_{{v'},i}$
	clusters of $\mathcal{T'}$.
	
	\item Look-up tables, which for a label of $\mathcal{T'}$
	allow to retrieve the corresponding cluster $C$.
\end{itemize}
\end{definition}

\begin{lemma}
\label{lem:cluster_structure}
For any labeled tree $\mathcal{T}$ and  $m$ we can construct in time $\Ocomp(\ts)$ cluster structure $C(\mathcal{T})$.
\end{lemma}

\section{Entropy estimation}\label{sec:entropy_estimation}
We show that entropy of labels of tree of $C(\mathcal{T})$
is upper bounded by mixed entropy of the input tree,
up to some small additive factor.
\begin{theorem}
\label{lem:main_estimation}
Let $\mathcal{T'}$ be a tree of structure $C(\T)$
from Lemma~\ref{lem:cluster_structure}
for parameter $m$,
obtained from $\mathcal{T}$.
Let $P$ be a string obtained by concatenation of labels of
$\mathcal{T'}$.
Then \emph{all} the following inequalities simultaneously hold:
\begin{align}
\label{main_thm:case1}
|P|H_0(P)
	&\leq
\ts H(\mathcal{T}) + \ts H_k(L) +
\Ocomp\left(\frac{\ts k \log \sigma}{m}
+ \frac{\ts \log m}{m} \right)\\
\label{main_thm:case2}
|P|H_0(P)
	&\leq
\ts H_k(\mathcal{T}|L) + \ts H_k(L)+
\Ocomp\left(\frac{\ts k \log \sigma}{m}
+ \frac{\ts \log m}{m} \right)\\
\label{main_thm:case3}
|P|H_0(P)
	&\leq
\ts H(\mathcal{T}) + \ts H_k(L|\mathcal{T})+
\Ocomp\left(\frac{\ts (k+1) \log \sigma}{m}
+ \frac{\ts \log m}{m} \right)\enspace .
\end{align}
Moreover, if $m = \Theta(\log_\sigma \ts)$,
the additional summands are bounded by:
$\Ocomp\left(\frac{\ts k \log \sigma}{\log_\sigma \ts} +
  \frac{\ts \log \log_\sigma \ts}{\log_\sigma \ts}
  \right)$
and
$\Ocomp\left(\frac{\ts(k+1) \log \sigma}{\log_\sigma \ts} +
\frac{\ts \log \log_\sigma \ts}{\log_\sigma \ts}
\right)$, 
respectively.
For $k = o(\log_\sigma \ts)$ both those values are $o(\ts \log \sigma)$.
\end{theorem}

We prove Theorem~\ref{lem:main_estimation} in two steps.
First, we devise a special representation of nodes of $\mathcal{T'}$.
The entropy of this representation is not larger than the entropy of $P$,
so we shall upper bound the entropy of the former.
To this end we use the following lemma,
which is simple corollary from Gibbs inequality,
see~\cite{aczel1973shannon} for a proof.
\begin{lemma}[\cite{aczel1973shannon}]\label{theoremP}
	Let $w \in \Gamma^*$ be a string and $q:\Gamma \rightarrow \mathbb{R}^+$ be a function such that $\sum_{s \in \Gamma} q(s) \leq 1$.
	Then
$
	|w|H_0(w) \leq -\sum_{s \in \Gamma} n_s \log q(s)
$,
	where $n_s$ is the number of occurrences of $s$ in $w$.
\end{lemma}
Lemma~\ref{theoremP} should be understood as follows:
we can assign each letter in a string a ``probability''
and calculate the ``entropy'' for a string using those ``probabilities.''
The obtained value is not smaller than the true empirical entropy.
Thus, to upper-bound the entropy, it is enough to devise an appropriate function $q$.

\paragraph{Cluster representation}
The desired property of the representation
is that for each node $v$ in the cluster we can
uniquely determine its context (i.e.\ labels of $k$ nodes on the path from $v$ to the root)
in original tree $\mathcal{T}$.
\begin{definition}[Cluster description]
	Given a cluster $C$ occurring in context $K$
	and consisting of subtrees $\mathcal{T}_1, \ldots, \mathcal{T}_l$,
	the \emph{cluster description}, denoted by $R_{K,C}$,
	is a triple $(K, N_C, V_C)$, where 
\begin{itemize}\setlength\itemsep{-0.2em}
	\item $K$ is a context (of size at most $k$) preceding
		roots of the trees in cluster, i.e.\ for each
		root $r$ of tree in a cluster $K = K_{r}$ holds,
		where $K_v$ is the context of vertex $v$ in $\mathcal{T}$;
		by the construction roots of trees in a cluster
		have the same context in $\T$ as they have the same parent.
	\item $N_C$ is the total number of nodes in this cluster;
	\item $V_C$ is a list of descriptions of nodes of $C$, according to preorder ordering.
	If a node $v$ is a port then its description is $(1, l_v)$,
	if it is a normal node, then it is $(0, l_v, d_v)$,
	where $d_v$ is the degree of the node $v$ (in the cluster)
	and $l_v$ is its label.
\end{itemize}
\end{definition}
Note that we do not store the degrees of port nodes, as they
are always $0$.

\begin{example}
	The description for $k=1$ and central node (the one which is labeled $B$ in cluster tree)
	from Figure~\ref{fig:trees} is:
	$(K=a, N = 4, V = \{ (1, b), (0, a, 2), (1, c), (1, a) \})$.
\end{example}

The following lemma says that the cluster description
satisfies desired properties:
\begin{lemma}
\label{lem:cluster_uniqely}
Cluster description uniquely defines a cluster
and context $K_v$ in $\mathcal{T}$ for each vertex $v$
\end{lemma}

To prove the Theorem~\ref{lem:main_estimation}
instead of estimating the entropy of $P$ we will estimate the entropy
of string where letters are descriptions of each cluster.
We assign each description value $q$ and use Lemma~\ref{theoremP}.
The values are assigned by assigning each element of the
description separate value of $q$ and multiplying it with previous
values for each element.
The value $q$ depends on previously encoded elements,
this method is similar to adaptive arithmetic coding.
Note that storing, which nodes are port
nodes inside a cluster description, uses additional $\Ocomp(\ts\log m / m)$ bits,
if done naively (e.g.\ in a separate structure) it would take the same memory.
The full proof is given in the Appendix.

\section{Application --- succinct data structure for labeled trees}
\label{sec:application_compressing}
We demonstrate how our tree clustering technique can be used
for indexing labeled trees.
Given a~tree $\T$ we choose $m = \Theta(\log_\sigma \ts)$
(the exact constant is determined later).
Then the number of different clusters is $\Ocomp(\ts^{1-\alpha})$ for some constant $\alpha > 0$,
thus cluster tree $\mathcal{T'}$ from $C(\T)$
 can be stored using
any succinct representation (like balanced parentheses or DFUDS)
in space $\Ocomp(\ts /\log_\sigma \ts) = o(\ts)$ (for small enough $\sigma$).
At the same time clusters are small enough so that
we can preprocess them and answer all relevant queries
within the clusters in constant time,
the needed space is also $o(\ts)$.

We construct the clustering $C(\mathcal{T})$
for $m = \beta \log_\sigma n$, where $\beta$ is a constant
to be determined later.
Let $\T'$ denote the unlabeled tree of $C(\mathcal{T})$,
i.e.\ the cluster tree~\Cref{1} stripped of node labels.
Our structure consists of:
\begin{enumerate}[(T1)]\setlength\itemsep{-0.2em}
	\item\label{T1} Unlabeled tree $\mathcal{T'}$, $|\mathcal{T'}| = \Ocomp(\ts / \log_\sigma \ts)$.
	\item\label{T2} String $P$ obtained by concatenating labels of
					the cluster tree of $C(\mathcal{T})$
					in preorder ordering.
	\item\label{T3} Degree sequences for each node of $\mathcal{T'}$.
	\item\label{T4} Precomputed arrays for each operation,
					for each cluster 
					(along with look-up table from $C(\mathcal{T})$ to decode
					cluster structure from labels).
\end{enumerate}
We encode each of \Trefall separately, using known tools.

\paragraph{\Tref{1}: Encoding tree $\mathcal{T'}$}
There are many succinct representation for unlabeled trees which allow
fast navigational queries, for example~\cite{fullyFunctionalTrees, succintTreesPractical,ultraSuccintTrees,farzan2009universal,geary2006succinct}.
Those methods use $2\tps + o(\tps)$ bits for tree of size $\tps$,
the exact function suppressed by the $o(\tps)$ depends on the data structure.
Since in our case the tree $\mathcal{T'}$ is already of size $\Ocomp(\ts/\log_\sigma \ts)$,
we can use $\Ocomp(\tps)$ bits for the encoding of $\T'$,
so we do not care about exact function hidden in $o(\tps)$.
This is of practical importance,
as the data structures with asymptotically smallest memory consumption,
like~\cite{fullyFunctionalTrees},
are very sophisticated, thus hard to implement and not always suitable for practical purposes.
Thus we can choose theoretically inferior,
but more more practical data structure~\cite{succintTreesPractical},
what is more we can even use a constant number of such data structures, as we are interested
only in $\Ocomp(\mathcal{|T'|})$ bound.

For the sake of the argument,
let us choose one method, say \cite{lu2008balanced}.
We use it to encode $\mathcal{T'}$ on $\Ocomp(\mathcal{|T'|})$ bits,
this encoding supports the following operations in constant time:
\procfont{parent($v$)} --- parent of $v$;
\procfont{firstchild($v$)} --- first child of a node $v$;
\procfont{nextsibling($v$)} --- right sibling of a node $v$;
\procfont{preorder-rank($v$)}  --- preorder rank of a node $v$;
\procfont{preorder-select($i$)} --- returns a node whose preorder rank is $i$;
\procfont{lca($u$, $v$)} --- returns the lowest common ancestor of $u, v$;
\procfont{childrank($v$)} --- number of siblings preceding a node $v$;
\procfont{child($v$, $j$)} --- $j$-th child of $v$;
\procfont{depth($v$)} --- distance from the root to $v$;
\procfont{level\_ancestor($v$, $i$)} --- ancestor at distance $i$ from $v$.

\paragraph{\Tref{2}: Encoding preorder sequence of labels} 
By Theorem~\ref{lem:main_estimation} it is enough
to encode the sequence $P$ using roughly $|P|H_0(P)$ bits,
in a way that allows for $\Ocomp(1)$ time access to its elements.
This problem was studied extensively,
and many (also practical) solutions were developed~\cite{GonzalezNStatistical,
FerraginaV07SimpStat,
GrossiDynamicIndexes,
grossi2003high}.
In the case of zero order entropy,
most of these methods are not overly complex:
they assign each element a prefix code,
concatenate prefix codes
and store some simple structure for storing
information, where the code words begin/end.

However,
we must take into account that alphabet of $|P|$ can be large
(though, as shown later, not larger than $\ts^{1-\alpha}$, $0 < \alpha < 1$).
This renders some of previous results inapplicable,
for example the simplest (and most practical) structure
for alphabet of size
$\sigma'$ need additional $\Ocomp(|P|\log \log |P| / \log_{\sigma'} |P|)$ which can be as large as $\Ocomp(|P|\log \log |P|)$~\cite{FerraginaV07SimpStat}.
As $|P| = \ts /  \log_\sigma \ts$ this would be
slightly above promised bound in Theorem~\ref{thm:theorem_tree_structure} 
($\Ocomp(\ts \log \log \ts / \log_\sigma \ts)$ vs $\Ocomp(\ts \log \log_\sigma \ts / \log_\sigma \ts)$).

Still, there are structures achieving $|P|H_0(P) + o(|P|)$ bits
for alphabets of size ${\ts}^{1-\alpha}$,
for example the well-known one by P\v{a}tra\c{s}cu~\cite{patrascu2008succincter}.

\paragraph{\Tref{3}: Encoding the degree sequence}
To navigate the tree we need to know which
children of a cluster belong to which port node.
We do it by storing degree sequence
for each cluster node of $\mathcal{T'}$
and design a structure which, given a node $u'$ of $\mathcal{T'}$
and index of a port node $u$ in the cluster represented by $u'$,
returns the range of children of $u'$
which contain children of $u$ in $\mathcal{T}$.
We do it by storing rank/select structure for bitvector representing
degrees of vertices of $\mathcal{T'}$ in unary.
\begin{lemma}
\label{lem:label_seq}
We can encode all degree sequences of nodes of $\mathcal{T'}$
using $\Ocomp (|\mathcal{T'}|)$ bits in total,
such that given node $u$ of $\mathcal{T'}$
and index of port node $v$ in cluster represented by $u$
the structure returns a pair of indices $i_1, i_2$,
such that the children of $p$ (in $\T$)
are exactly the roots of trees in clusters
in children $i_1, i_1+1, \ldots, i_2-1$ of $u$.
Moreover we can answer reverse queries,
that is, given an index $x$ of $x$-th child of 
$u \in \mathcal{T'}$ find
port node which connects to this child. 
Both operations take $\Ocomp(1)$ time.
\end{lemma}

\paragraph{\Tref{4}: Precomputed tables}
We want to answer all queries in each cluster in constant time,
we start by showing that there are not many clusters of given size.
\begin{lemma}\label{ref:lem_size_clusters}
There are at most $2^{2m'} \sigma^{m'}$
different clusters of size $m'$.
\end{lemma}
It is sufficient to choose $m$ in clustering
as $m=\frac{1}{8}\log_\sigma \ts$ (or $1$ if this is smaller than $1$)
assuming that $2 \leq \sigma \leq \ts^{1-\alpha}$, $\alpha > 0$,
then number of different
clusters of length at most $2m-1$ is
$
\sum_{i=1}^{2m-1} 2^{2i} \sigma^{i} \leq \Ocomp(n^{1-\alpha})
$.

We precompute and store the answers for each query for each cluster.
As every query takes constant number of arguments
and each argument ranges over $m$ values,
this takes at most $\Ocomp(\ts^{1-\alpha}) \cdot \Ocomp(\log_\sigma^c \ts) = o(\ts)$ bits,
where $c$ is a constant.

Additionally we make tables for accessing $i$-th
(in left-to-right order on leaves)
port node of each cluster, 
as we will need this later to navigate more involved queries.

\paragraph{Putting it all together}
We now show that above structures can be combined into
a succinct data structure for trees.
Note that previously we used the fact that $\sigma \leq \ts^{1-\alpha}$,
in appendix we generalize for a case $\sigma = \omega(\ts^{1-\alpha})$.

The main idea of Theorem~\ref{thm:theorem_tree_structure} is that 
for each query if both arguments and the answer are in the same cluster 
then we can use precomputed tables,
for other we can query the structure for $\mathcal{T'}$ and reduce it to the former
case using previously defined structures.
Similar idea was used in other tree partition based structures like~\cite{geary2006succinct}
(yet this solution used different tree partition method).
\begin{theorem}\label{thm:theorem_tree_structure}
Let $\mathcal{T}$ be a labeled tree with labels
from an alphabet of size $\sigma \leq \ts^{1-\alpha}, \alpha > 0$.
Then we can build the tree structure which consumes the number of bits bounded by 
\emph{all} of the chosen value from the list below:
\begin{align*}
&\ts H(\mathcal{T}) + \ts H_k(L) +
	\Ocomp\left(\frac{\ts k\log \sigma}{\log_\sigma \ts}
	+ \frac{\ts \log \log_\sigma \ts}{\log_\sigma \ts} \right)\\
&	\ts H_k(\mathcal{T}|L) + \ts H_k(L)+
	\Ocomp\left(\frac{\ts k \log \sigma}{\log_\sigma \ts}
	+ \frac{\ts \log \log_\sigma \ts}{\log_\sigma \ts} \right)\\
&	\ts H(\mathcal{T}) + \ts H_k(L|\mathcal{T})+
	\Ocomp\left(\frac{\ts (k+1) \log \sigma}{\log_\sigma \ts}
	+ \frac{\ts \log \log_\sigma \ts}{\log_\sigma \ts} \right)
\end{align*}
It supports 
\procfont{firstchild(u)}, \procfont{parent(u)},
\procfont{nextsibling(u)}, \procfont{lca(u,v)},
\procfont{childrank($u$)}, \procfont{child($u$,$i$)} and \procfont{depth($u$)}
 operations in $\Ocomp(1)$ time;
 moreover
 with additional $\Ocomp(\ts (\log \log \ts)^2 / \log_\sigma \ts)$ bits 
 it can support \procfont{level\_ancestor($v$, $i$)} query.
\end{theorem}

\section{Even succincter structure}\label{sec:even_succinter}
So far we obtained the redundancy $\Ocomp(\ts\log \log_\sigma \ts / \log_\sigma \ts) + 
\Ocomp(\ts k\log\sigma /\log_\sigma \ts)$.
%
As the recent lower bound for zeroth-order entropy coding string partition~\cite{entropyBoundsUnpublished}
(assuming certain partition properties)
also applies to trees, we can conclude
that our structure in worst case requires $\Ocomp(\ts k\log\sigma /\log_\sigma \ts)$ additional bits.
Yet, this lower bound only says the above factor is necessary
when zeroth-order entropy coder is used,
not that this factor is required in general.
Indeed, for strings there are methods of compressing
the text $S$ (with fast random access) using $|S|H_k(S) + o(|S|) + f(k, \sigma)$ bits~\cite{GrossiDynamicIndexes, munro2015compressed}
also methods related to compression boosting 
achieve redundancy $|S|H_k(S) + \Ocomp(|S|) + f(k, \sigma) $ bits~\cite{FerraginaCompressionBoosting},
where $f(k, \sigma) $ is some function that depends only on $k$ and $\sigma$.
Similarly, the method of
compressing the tree using combination of \algofont{XBWT}
and compression boosting gives redundancy of $\Ocomp(|S|) + f(k, \sigma)$ bits~\cite{ferragina2005labeledTrees}.
In all of the above cases $f(k, \sigma)$ can be bounded by $\Ocomp(\sigma^k \cdot \text{polylog}(\sigma))$.
This is more desirable than $\Ocomp(\ts k\log\sigma /\log_\sigma \ts)$,
as in many applications $k$ and $\sigma$ are fixed and so this term is constant,
moreover the redundancy is a sum of two functions instead of a product.
Moreover, achieving such redundancy allows us to relax our assumptions,
i.e.\ we obtain additional $o(n)$ factor for $k = \alpha \log_\sigma n, 0 < \alpha < 1$,
while so far our methods only gave $o(n \log \sigma)$ for $k = o(\log_\sigma n)$.

We show how to decrease the redundancy to $\Ocomp(\sigma^k \cdot \text{polylog}(\sigma))$
at the cost of increasing the query time to $\Ocomp(\log n / \log \log n)$
(note that previously mentioned compressed text storages~\cite{GrossiDynamicIndexes} also did not support constant access).
The proof of Theorem~\ref{lem:main_estimation} suggests that we lose up to $k \log \sigma $
bits per cluster
(as a remainder: we assign each cluster value $q$ and we ``pay'' $\log q$ bits),
so it seems like a bottleneck of our solution.
In the case of text compression~\cite{GrossiDynamicIndexes} improvements were obtained by
partitioning the text into blocks and encoding string made of blocks of sizes $\log_\sigma n$ with first order entropy coder,
to support retrieval in time $\Ocomp(d)$, for some $d$, every $d$-th block
was stored explicitly.
For $d$ = $\log |S| / \log \log |S|$ and assuming that each block has at most $\log_\sigma |S|$
characters this gives $\Ocomp(|S| \log \log |S| / \log_\sigma |S|)$ bits of  redundancy.

We would like to generalize this idea to labeled trees,
yet there are two difficulties: first, the previously mentioned
solution for strings needed the property that context of each block is stored wholly in some previous block,
second, as there is no linear order on clusters, we do not know how to choose $\tps/d$ clusters.
Our approach for solving this problem combines ideas from both
compression boosting techniques~\cite{FerraginaCompressionBoosting}
and for compressed text representation~\cite{GrossiDynamicIndexes}:
The idea is that for $\Ocomp(\tps / d)$ nodes, where $d = \log \ts / \log \log \ts$,
we store the context explicitly using $k\log\sigma$ bits.
The selection of nodes is simple:
by counting argument for some $1 \leq i < d$ the tree levels $i, i+d, \ldots i+dj$ of the cluster tree $\mathcal T'$
have at most $\tps / d$ nodes.
This allows to retrieve context of each cluster by traversing
(first up and then down) at most $d$ nodes.
Then we partiton the clusters in the classes depending on their preceding context
and use zeroth order entropy for each class
(similarly to compression boosting~\cite{FerraginaCompressionBoosting}
or some text storage methods~\cite{GonzalezNStatistical}),
i.e.\ we encode each cluster as if we knew its preceding context.
The decoding is done by first retrieving the context
and next decoding zeroth order code from given class.

\begin{theorem}\label{thm:theorem_tree_structure_boosting}
	Let $\mathcal{T}$ be a labeled tree and $k = \alpha \log_\sigma \ts, \alpha < 1$.
	Then we can build the structure using one chosen number of bits from the list below:
	\begin{itemize}\setlength\itemsep{-0.2em}
		\item $\ts H(\mathcal{T}) + \ts H_k(L) +
		\Ocomp\left(
		\frac{\ts \log \log \ts}{\log_\sigma \ts} \right)$;
		\item $\ts H_k(\mathcal{T}|L) + \ts H_k(L)+
		\Ocomp\left(
		\frac{\ts \log \log \ts}{\log_\sigma \ts} \right)$;
		\item $\ts H(\mathcal{T}) + \ts H_k(L|\mathcal{T})+
		\Ocomp\left(
		\frac{\ts \log \sigma}{\log_\sigma \ts}
		+ \frac{\ts \log \log \ts}{\log_\sigma \ts} \right)$.
	\end{itemize}
	This data structure supports the following operations
	in $\Ocomp(\log \ts / \log \log \ts)$ time:
	\procfont{firstchild($u$)}, \procfont{parent($u$)},
	\procfont{nextsibling($u$)}, \procfont{lca($u,v$)},
	\procfont{child($v,i$)}, \procfont{childrank($u$)},
	\procfont{depth($v$)}.
	Using additional $\Ocomp(\ts {(\log \log \ts)}^2 / \log_\sigma \ts)$
	bits it can support \procfont{level\_ancestor($v,i$)}
	query.
\end{theorem}
\section{Label-related operations}\label{sec:label_operations}
We show that using additional memory we can support
some label-related operations
previously considered for succinct trees~\cite{tsur2015succinct,he2014framework}.
Even though we do not support all of them, in all cases we support at least the same operations
as \algofont{XBWT}, i.e.\
\procfont{childrank($v, a$)} (which returns $v$'s rank among children labeled with $a$)
and \procfont{childselect($v, i, a$)} (which returns the $i$th child of $v$ labeled with $a$).
To this end we employ \procfont{rank/select} structures for large alphabets~\cite{BarbayRankSelect,OptimalLowerUpper}.
Note, that we do not have constant time for every alphabet size
and required additional space is larger than for other operations (i.e.\ $o(n\log \sigma)$),
but this is unavoidable~\cite{OptimalLowerUpper, GrossiOptimalTradeoffIndexes}.
Moreover, as the last point of Theorem~\ref{thm:label_operations}
use structures which are not state-of-the-art,
it is likely that better structures~\cite{OptimalLowerUpper} are applicable,
but they are more involved and it is not clear if
they are compatible with our tree storage methods.

\begin{theorem}
	\label{thm:label_operations}
	We can augment structures from Theorem~\ref{thm:theorem_tree_structure} and
	Theorem~\ref{thm:theorem_tree_structure_boosting} so that:
	\begin{itemize}\setlength\itemsep{-0.2em}
		\item for $\sigma = \Ocomp(1)$ we can perform:
		\procfont{childrank($v, a$)}, \procfont{childselect($v, i,a$)},
		\procfont{level\_ancestor($v, i, a$)}, \procfont{depth($v, a$)}
		for structure from Theorem~\ref{thm:theorem_tree_structure}
		in $\Ocomp(1)$ time and
		for structure from Theorem~\ref{thm:theorem_tree_structure_boosting}
		in $\Ocomp(\log \ts / \log \log \ts)$ time
		using asymptotically the same additional memory.
		\item for $\sigma = \Ocomp(\log^{1+o(1)}) \ts$ and $\sigma = \omega(1)$ we can perform:
		\procfont{childrank($v, a$)} and \procfont{childselect($v, i, a$)}
		for structure from Theorem~\ref{thm:theorem_tree_structure}
		in $\Ocomp(1)$ time and
		for structure from Theorem~\ref{thm:theorem_tree_structure_boosting}
		in $\Ocomp(\log \ts / \log \log \ts)$ time
		using additional $o(\ts \log \sigma)$ bits.
		\item for arbitrary $\sigma$ we can perform:
		\procfont{childrank($v, a$)} and \procfont{childselect($v, i, a$)}
		for structure from Theorem~\ref{thm:theorem_tree_structure}
		in $\Ocomp( \log \log^{1+\epsilon} \sigma )$ time and
		for structure from Theorem~\ref{thm:theorem_tree_structure_boosting}
		in $\Ocomp((\log \log^{1+\epsilon} \sigma) \log \ts / \log \log \ts) $ time;
		using additional $o(\ts \log \sigma)$ bits.
	\end{itemize}
\end{theorem}

\section{Conclusions and future work}
We have shown a simple structure allowing
fast operations on the compressed form.
In some cases it is better than previous results,
moreover, it is the first compression
method which the exploits correlation
between tree structure and labels
and achieves theoretical bounds.

The open problems are in the very first section of the Appendix.

\newpage
\appendix
\begin{center}
\Huge Appendix
\end{center}

\section{Open problems}

There are a few open questions.
First, can our analysis be applied to
recently developed dictionary compression methods
for trees like Top-Trees/Top-Dags
~\cite{bille2013TopTrees,hubschle2015tree}
or other dictionary based methods on trees?
Related is the problem of finding good
compression measures for repetitive trees,
for instance for the case of text we have
\algofont{LZ77} and \algofont{BWT}-run,
for which we can build efficient structures
(like text indicies)
based on this representations~\cite{selfIndexLz77, OptimalBWTIndexing}
\emph{and} find relation with information-theoretic bounds (like $k$-th order entropy)
(or even show that they are close in information-theoretic sense to each other~\cite{StringAttractorsSTOC}).
Even though we have tree representation like \algofont{LZ77} for trees~\cite{treeLZ}
or tree grammars~\cite{treegrammar} we do not know relation
between them nor how they correspond to tree entropy measures.
One could also measure tree repetitiveness with number of  runs
(i.e.\ number of phrases in run-length encoding)
in string generated by $\algofont{XBWT}$, yet, as mentioned before,
it does not capture tree shape and to our knowledge
such measure was not considered before;
so there one problem is improving such measure,
the other is to construct data structure achieving such quantity.

Second, we do not support all of the labeled operations,
moreover we do not achieve optimal query times.
The main challenge is to support more complex operations,
like labeled level\_ancestor, while achieving theoretical
bounds considered in this work.
Previous approaches partitioned (in rather complex way) the tree into
subtrees, by node labels (i.e.\ one subtree contained only nodes with same label)
and achieved at most
zeroth-order entropy of labels~\cite{tsur2015succinct,he2014framework},
it maybe possible to combine these methods with ours.

Next, it should be possible that the presented structure
can be made to support dynamic trees,
as all of the used structures have their dynamic
equivalent~\cite{GrossiDynamicIndexes,succintPartialSums,fullyFunctionalTrees}.
Still, it is not entirely trivial as we need to
maintain clustering.

Next, as Ferragina et al.~\cite{ferragina2005labeledTrees}
mentioned in some applications nodes can store strings,
rather than single labels,
where the context for a letter is defined by labels of ancestor
nodes and previous letters in a node.
It seems that our method should apply in this scenario,
contrary to \algofont{XBWT}.

At last,
can the additional space required for \procfont{level\_ancestor}
query can be lowered?
We believe that this is the case,
as our solution did not use the fact
that all weights sum up to n;
moreover it should be possible to apply methods
from from~\cite{fullyFunctionalTrees}
to obtain $\Ocomp(\ts \log \log \ts / \log_\sigma \ts)$
additional space for \procfont{level\_ancestor}.

\section{Additional proofs for Section~\ref{sec:definitions}}
\label{sec:definitions_appendix}
\begin{proof}[Proof of Lemma~\ref{lem:entropies_not_larger}]
The proof follows by straightforward application of the log sum inequality:
\begin{align*}
 \ts H_k(\mathcal{T}|L) &=
 -\sum_{v \in \mathcal{T}} \log \prob (d_v | K_{v}, l_v)\\
 &= -\sum_{v \in \mathcal{T}} \log \frac{t_{K_v, l_v, d_v}}{t_{K, l_v}}\\
 &= -\sum_{d} 
	\sum_{K}
	\sum_{l}
 t_{K, l, d}\log \frac{t_{K, l, d}}{t_{K, d}}\\
 &\leq -\sum_{d}  t_d \log \frac{t_d}{\ts}
 \end{align*}
The proof for the case $\ts H_k(L|\mathcal{T})$ is similar.
\end{proof}

\section{Additional proofs for Section~\ref{sec:tree_clustsering}}

\begin{proof}[Proof of Lemma~\ref{lem:cluster_structure}]
	Given a~tree $\T$ we cluster its nodes,
	and create a node for each cluster
	and add an edge between two nodes $u$ and $v$
	if and only if in $\mathcal{T}$ there was an edge from some vertex
	from cluster $C_u$ to some vertex in $C_v$.
	We label the cluster nodes consistently, i.e.\
	$u$ and $v$ get the same label if and only if their clusters
	are identical,
	also in the sense which nodes are port nodes
	(note that the port nodes can have different degree in the input tree).
	For each cluster label we store its cluster.
	For simplicity we assume that assigned labels are 
	from the ordered set (i.e.\ set of numbers).
	
	As mentioned before we store the previously defined degree sequence.
\end{proof} 

\begin{proof}[Proof of Lemma~\ref{lem:clustering_procedure}]
	We build the clustering by a simple dfs-based method,
	starting at the root $r$.

	For a node $v$, the procedure returns a tree $c_v$ rooted in $v$
	along with its size (and also creates some clusters).
	The actions on $v$ are as follows:
	we recursively call the procedure on $v$'s children \seq v 1 j,
	let the returned trees be \seq c 1 j and their sizes $\seq s 1 j$.
	We have two possibilities:
	\begin{itemize}\setlength\itemsep{-0.2em}
		\item $\sum_{i=1}^j s_i < m$, then return a tree rooted at $v$ 
		with all of the returned trees \seq c 1 j rooted at its children.
		\item $\sum_{i=1}^j s_i \geq m$,
		then we group returned trees greedily:
		We process trees from left to right, at the beginning we create cluster containing $C=\{c_1\}$,
		while $|C| < m$ we add consecutive $c_i$'s to $C$
		(recall that $|C|$ denotes the number of nodes in trees in $C$).
		Then at some point we must add $c_j$
		such that $|C| \geq m$.
		We output the current cluster $C$, set $C= \{ c_{j+1}\}$ and continue grouping the trees.
		At the end we return tree containing just one vertex: $v$.
	\end{itemize}
	FinallyNote that we make a cluster of the tree returned by the root,
	regardless of its size.

	First observe that indeed this procedure always returns a tree,
	its size is at most $m$: either it is only a root or in the first case the sum of sizes of subtrees is at most $m-1$, plus $1$ for the root.
	
	This implies \Cref{2}: when we make a cluster out of $\seq c i {j+1}$ then $\sum_{\ell=i}^j s_\ell$ is at most $m-1$ by the algorithm
	and $s_{j+1} \leq m$ by the earlier observation.
	
	It is easy to see from the algorithm that a cluster contains trees rooted at the consecutive siblings, so \Cref{3} holds.
	
	By an easy induction we can also show that in each returned tree
	the degree of the node is $0$ (for the root of the tree in the second case)
	or equal to the degree in the input tree (the root in the first case),
	thus \Cref{4} holds.
	
	Concerning the total number of clusters,
	first recall that they are of size at most $2m-1$,
	thus there are at most $\frac{\ts}{2m-1} \geq \frac{\ts}{2m}-1$ clusters.
	To upper bound the number of clusters observe
	that by the construction the clusters are of size at least $m$ except two cases:
	the cluster rooted at the root of the whole tree
	and the clusters that include the last child of the root (but not the root, i.e.\ they are created in the second case).
	For the former, there is at most one such a cluster.
	For the latter, in this case we also created at least one other cluster from trees rooted in siblings in this case,
	its size is at least $m$.
	Thus there are at most $2\frac{\ts}{m}+1$ clusters, as claimed in \Cref{1}.
	%
	%
	%
	%
	%
	%
	%
\end{proof}

\begin{proof}[Full version of proof of Theorem~\ref{lem:main_estimation}]
	Let $P$ be a sequence of labels over alphabet $\Gamma_{L}$.
	Let $R$ be a sequence of description of clusters of $\mathcal{T'}$
	in preorder ordering, so that description $R[i]$ of cluster $C$
	corresponds to label $P[i]$. 
	Observe that, $H_0(P) \leq H_0(R)$,
	in particular $\en{P} \leq \en{R}$:
	To see this, observe that each label of $P$ may
	correspond to many descriptions from $R$,
	but not the other way around:
	if some nodes have different descriptions,
	then they have different labels in $\mathcal{T'}$.

	Thus it is enough to upper-bound \en{R}.
	This is done by applying Lemma~\ref{theoremP}
	for appropriately defined function $q$;
	different estimations require different variants of $q$.
	The function $q$ is defined on each $R[i]$.

	We begin with a proof for Case~\ref{main_thm:case2}.
	Let $R_{K,C} = (K, N_C, V_C)$ be a description of the cluster $C$.
	We define $q(R_{K,C}) = q(K_C)\cdot q(N_C) \cdot q(V_C)$,
	where $q$ for each coordinate is defined as follows:
	\begin{itemize}\setlength\itemsep{0.1em}
		\item $q(K_C) = 1/\left((|K_C|+1)
		\cdot \sigma^{|K_C|}\right)$
		\item $q(N_C) = 1/(2m)$
		\item $q(V_C) = \prod_{v \in V_C} q(v)$,
		where 
		\begin{equation*}
		q(v) = \begin{cases}
		\frac{1}{m}\cdot
		\prob(l_v | K_v)
		&\text{, if } v = (1, l_v)	\\
		\frac{m-1}{m} \cdot
		\prob(l_v | K_v) \cdot 
		\prob(d_v | K_v, l_v)
		&\text{, if } v = (0, l_v, d_v)
		\end{cases}
		\end{equation*}
	\end{itemize}	
	Note that $K_v$ is the context of $v$ in original tree, i.e.\ $\mathcal{T}$.
	
	It is left to show that the $q$ summed over all cluster descriptions is at most $1$.
	To this end we will show that we can partition the interval $[0,1]$ into subintervals
	and assign each element of $\Gamma_{L}$
	subinterval of length $q(R_C)$,
	such that for two symbols of $\Gamma_L$
	their intervals are pairwise disjoint.
	It is analogous to applying adaptive arithmetic
	coder.
	We start with interval $I = [0, 1]$.
	We process description of cluster by coordinates,
	at each step we partition $I$ into disjoint
	subintervals and choose one as the new $I$:
	\begin{itemize}\setlength\itemsep{0.1em}
		\item For $K_C$ we partition interval into $|k|+1$
		equal subintervals, each one corresponding to
		different context length,
		then we choose one corresponding to $|K_C|$.
		Then we partition the $I$ into
		$\sigma^{|K_C|}$ disjoint and equal subintervals,
		one for each different context of length $|K_C|$
		and choose one corresponding to $K_C$.
		Clearly the length of the current $I$ at this point
		is $1/\left((|K_C|+1)\cdot \sigma^{|K_C|}\right)$,
		also different contexts are assigned different intervals.
		\item For $N_C$ we partition the $I$ to $2m$ equal intervals,
		it is enough, as there are at most $2m$ different cluster sizes.
		\item Then we make a partition for each $u \in V_C$.
		We process vertices of $V_C$ as in they occur on this list
		(i.e.in preorder ordering).
		We partition the interval into two, one of length
		$\frac{|I|}{m}$, second $\frac{|I|(m-1)}{m}$.
		If $u = (1, l_v)$ we choose the first one,
		otherwise we choose the second one.
		Then we partition the $|I|$ into
		$\sigma$ different intervals (but some may be of $0$ length),
		one for each letter $a$ of original alphabet $\Sigma$;
		the subinterval for letter $a$ has length $|I| \cdot \prob(a | K_C)$.
		It is proper partition, as all of the above
		values sum up to $1$ by definition.
		We choose the one corresponding to the letter $l_v$.
		Lastly, if $i = (0, l_v, d_v)$ then
		we partition the interval into intervals
		corresponding to different degrees of nodes,
		again of lengths $|I|\cdot\prob(d_v | K_C, l_v)$.
	\end{itemize}
	
	By construction the interval assigned to $R_C$ has length $q(R_C)$.
	Also, for two different clusters $C_1, C_2$, having different preceding
	contexts $K_{C_1}, K_{C_2}$,
	their intervals are disjoint.
	To see this consider the above procedure which assigns intervals
	and the
	first point where descriptions $R_{K_{C_1}, C_1}$,
	$R_{K_{C2}, C_2}$ differ (there must be such point by Lemma~\ref{lem:cluster_uniqely}).
	Observe that up to this point both clusters were assigned the same intervals.
	What is more already processed elements of $R_{K, C}$
	uniquely define context (or label) for current vertex,
	i.e.\ when we want to partition by $\prob(l_v | K_v)$,
	the all nodes on the path to $v$ were already processed
	(this is due to the preorder ordering
	of nodes in description).
	This guarantees that, if the descriptions
	for two clusters $C_1, C_2$
	were equal up to this point,
	then the interval will be partitioned in the same way for $C_1$ and $C_2$.
	At this point $R_{C_1}$, $R_{C_2}$ will be assigned different,
	disjoint intervals. 
	
	Now we are ready to apply Lemma~\ref{theoremP}.
	By $C_{v}$ and $K_{v}$ we denote the 
	cluster of represented by $v$ and context of 
	this cluster in $\mathcal{T}$, respectively;
	additionally let $n=\ts$.
	\begin{align*}
	|P|H_0(P) &\leq
	|R|H_0(R) \\&\leq
	-\sum_{v \in \mathcal{T'}} \log q(R_{K_{v}, C_{v}})\\
	&= -\sum_{v \in \mathcal{T'}}
	\log \left(q(K_v) \cdot q(N_{C_v}) \cdot q(V_{C_v}) \right) \\	
	&\leq-\sum_{v \in \mathcal{T'}}
	\log \left(\frac{1}{(k+1)\sigma^k} \cdot \frac{1}{m} \cdot q(V_{C_v}) \right)\\
	&=|\mathcal{T'}|\left(k\log\sigma + \log (k+1) + \log m \right)
	-\sum_{v \in \mathcal{T'}}
	\log \left(q(V_{C_v}) \right)\\
	&\leq -\sum_{v \in \mathcal{T'}}\log \left(q(V_{C_v}) \right) +
	\Ocomp\left(\frac{n k \log \sigma }{m} + \frac{n\log m}{m} \right) \enspace .
	\end{align*}
	Now:
	\begin{align*}
	-\sum_{v \in \mathcal{T'}}\log q(V_{C_v}) &= 
	-\sum_{u \in \mathcal{T}}\log q(u)\\
	&=-\sum_{\substack{u \in \mathcal{T}\\u: \text{port}} }\log\left(\frac{1}{m} \cdot \prob(l_v | K_v) \right)
	-\sum_{\substack{u \in \mathcal{T}\\u: \text{normal}}}
	\log
	\left(
	\frac{m-1}{m} \cdot
	\prob(l_v | K_v) \cdot 
	\prob(d_v | K_v, l_v) \right)\\
	&\leq nH_k(L) + nH_k(\mathcal{T|L}) + \Ocomp(|\mathcal{T'}|\log m) + n \log \frac{m}{m-1}\\
	&\leq 
	nH_k(L) + nH_k(\mathcal{T|L}) + \Ocomp\left(\frac{n \log m}{m}\right)+
	\frac{n}{m-1}\log \left(1 + \frac{1}{m-1} \right)^{m-1}\\
	&\leq
	nH_k(L) + nH_k(\mathcal{T|L}) + \Ocomp\left(\frac{n \log m}{m}\right)+
	\Ocomp\left(\frac{n}{m} \right) \enspace ,
	\end{align*}
	which ends the proof for the Case~\ref{main_thm:case2}.
	In the estimation we have used the fact that the total number of port nodes
	is at most $\Ocomp(n/m)$, since it cannot exceed the number of clusters.
	
	The proof of Case~\ref{main_thm:case1} can be carried out in a similar manner,
	by replacing $\prob(d_v | K_v, l_v)$ with $\prob(d_v | K_v)$;
	alternatively it follows from Lemma~\ref{lem:entropies_not_larger}.
	
	The Case~\ref{main_thm:case3} requires slight modification of assignment
	of $q$ to vertices:
	\begin{equation*}
	q(v) = \begin{cases}
	\frac{1}{m}\cdot
	\frac{1}{\sigma}
	&\text{, if } v = (1, l_v)	\\
	\frac{m-1}{m} \cdot
	\prob(d_v | K_v) \cdot 
	\prob(l_v | K_v, d_v)
	&\text{, if } v = (0, l_v, d_v)
	\end{cases} \enspace .
	\end{equation*}
	This is because we do not have
	information on original degree of $v$,
	so we cannot partition the interval by
	$\prob(l_v | K_v, d_v)$ for port nodes:
	we want to keep the invariant that when
	partitioning the interval for $\prob(l_v | K_v, d_v)$
	previous elements of $R_{K, C}$ uniquely determins $K_v, d_v$,
	which may not be true.
	This adds additional $\Ocomp\left(\frac{n \log \sigma }{m}\right)$
	summand (since there are at most $\Ocomp(n/m)$ port nodes), hence we have
	$\Ocomp\left(\frac{n (k+1) \log \sigma}{m}\right)$
	in the third case.
	
	Also the argument with assigning intervals must be slightly changed:
	we first partition the interval by $\prob(d_v | K_v)$,
	then by $\prob(l_v | K_v, d_v)$.
\end{proof}

\section{Additional proofs for Section~\ref{sec:entropy_estimation}}
\begin{proof} [Proof of Lemma~\ref{lem:cluster_uniqely}]
	It is a known fact that from sequence of degrees
	in preorder ordering
	we can retrieve shape of the tree (we can do it by 
	simple dfs-procedure, which first creates node with given degree,
	then calls itself recursively; when we recurse back we know
	which node is next etc.).
	From cluster description we can easily retrieve the
	sequence of its degrees in preorder ordering.
	We also can easily retrieve labels and information which nodes are port.
	Now for each vertex $v$ in cluster we know its original context $K_v$ in $T$,
	as we explicitly store context for roots of trees,
	and other nodes have their context either fully in cluster,
	or their context is concatenation of some suffix of $K$
	and some path in the cluster.
\end{proof}

\section{Additional material for Section~\ref{sec:application_compressing}}

\begin{proof}[Proof of Lemma~\ref{lem:label_seq}]
	First we concatenate degree sequences
	for each node in $\mathcal{T'}$,
	according to preorder ordering,
	obtaining a sequence
	$D = d_{{v_1},1}, \ldots d_{{v_1},j_1}, d_{{v_2},1}, \ldots d_{{v-2},j_2},
	d_{{v_\mathcal{|T'|}},1}, \ldots d_{{v_\mathcal{|T'|}},j_\mathcal{|T'|}}$.
	Sum of all $d_{v,j}$'s in the sequence is bounded by $|\mathcal{T'}|$,
	as each $d_{v,j}$ corresponds to $d_{v,j}$ edges.
	We now encode each number in the sequence in unary:
	$D_u= 0^{d_{{v_1},1}}1 \ldots 0^{d_{{v_\mathcal{|T'|}},j_\mathcal{|T'|}}}1$
	
	Then we build a separate sequence $B_u$,
	which marks the borders between nodes in the degree sequence,
	i.e $B_u[z] = 1$ if and only if at index $z$
	starts the unary degree sequence of some node.
	
	Consider the following example: for nodes $a,b,c,d,e$
	and corresponding degree sequences $(0), (3,1), (2), (1,2), (2,2)$
	we have
	(the vertical lines $|$ denote borders of degree sequences
	and are added for increased readability):
	\begin{align*}
	D_u &= 1|000101|001|01001|001001 \\
	B_u &= 1|100000|100|10000|100000
	\end{align*}
	We now construct rank/select data structure for $D_u$ and $B_u$.
	There are multiple approaches that, for static bitvectors,
	use $\Ocomp(|D_u| + |B_u|) = \Ocomp(|\mathcal{T'}|)$ bits
	and allow both operations in $\Ocomp(1)$
	time~\cite{succintDictionariesWithTrees}.

	We describe how to answer a query for node $u \in \mathcal{T'}$
	and port node $v$ in the cluster.
	Let $p_u$ be preorder index of $u$.
	
	Let $j$ be the point in $D_u$ where the degree sequence of $u$
	starts, i.e.\ $j = select_1(p_u, B_u)$.
	Let $p_v$ be a port index of $v$ in the $u$-cluster
	(recall that we ordered port nodes in each cluster in left-to-right order on leaves).
	We want to find the beginning of unary description of $d_{u, p_v}$
	(plus one) in $D_u$:
	this is the $p_v-1$-th $1$ starting from $j$-th element in $D_u$.
	The next $1$ corresponds to the end of unary description.
	Let $u_1, u_2$ be the beginning and end of this unary description,
	we can find them in the following way:
	$u_1 = select_1(rank_1(j, D_u) + p_v - 1, B_u) + 1$
	and  $u_2 = select_1(rank_1(j, D_u) + p_v, B_u) -1$.
	Observe now that number of $0$'s in $D_u$
	between $j$ and $u_1$ (with $j$ and $u_1$)
	is equal to $i_1$.
	Similarly we can get $i_2$ by counting zeroes between $u_1$ and $u_2$.
	Thus $i_1 = rank_0(u_1, D_u) - rank_0(j)$ and $i_2 = rank_0(u_2, D_u) - rank_0(u_1, D_u)-1$.
	
	We now proceed to the second query.
	We find the beginning of description in unary,
	denoted $j$ as above.
	We find position $u_x$ of $x$-th $0$
	counting from $j$,
	we do it by calling $u_x =  select_0(rank_0(j) + x)$.
	Now we calculate the numbers of $1$'s between $j$ and $u_x$
	and simply return this value ($+1$).
	That is, we return $rank_1(u_x) - rank_1(j) + 1$.
\end{proof}

\begin{proof}[Proof of Lemma~\ref{ref:lem_size_clusters}]
	Each cluster can be represented by:
	a number of nodes in cluster, written as a~unary string of length $m'+1$,
	a bitvector indicating which nodes are port nodes of length $m'$,
	balanced parentheses representation of cluster structure,
	string of labels of length $m'$.
\end{proof}

We give more general version of Theorem~\ref{thm:theorem_tree_structure}.
\begin{theorem}[Full version of Theorem ~\ref{thm:theorem_tree_structure}]
	\label{thm:theorem_tree_structure_full}
	Let $\mathcal{T}$ be a labeled tree with labels
	from an alphabet of size $\sigma \leq \ts^{1-\alpha}, \alpha > 0$.
	Then we can build the tree structure
	using one chosen number of bits from the list below:
	\begin{itemize}\setlength\itemsep{0.1em}
		\item $\ts H(\mathcal{T}) + \ts H_k(L) +
		\Ocomp\left(\frac{\ts k\log \sigma}{\log_\sigma \ts}
		+ \frac{\ts \log \log_\sigma \ts}{\log_\sigma \ts} \right)$;
		\item $\ts H_k(\mathcal{T}|L) + \ts H_k(L)+
		\Ocomp\left(\frac{\ts k \log \sigma}{\log_\sigma \ts}
		+ \frac{\ts \log \log_\sigma \ts}{\log_\sigma \ts} \right)$;
		\item $\ts H(\mathcal{T}) + \ts H_k(L|\mathcal{T})+
		\Ocomp\left(\frac{\ts (k+1) \log \sigma}{\log_\sigma \ts}
		+ \frac{\ts \log \log_\sigma \ts}{\log_\sigma \ts} \right)$.
	\end{itemize}
	
	For general $\sigma$ we can build the structure that consumes:
	\begin{itemize}\setlength\itemsep{0.1em}
		\item $\ts H(\mathcal{T}) + \ts H_k(L) +
		\Ocomp\left(\frac{\ts k\log \sigma}{\log_\sigma \ts}
		+ \frac{\ts \log \log \ts}{\log_\sigma \ts} \right)$;
		\item $\ts H_k(\mathcal{T}|L) + \ts H_k(L)+
		\Ocomp\left(\frac{\ts k \log \sigma}{\log_\sigma \ts}
		+ \frac{\ts \log \log \ts}{\log_\sigma \ts} \right)$;
		\item $\ts H(\mathcal{T}) + \ts H_k(L|\mathcal{T})+
		\Ocomp\left(\frac{\ts (k+1) \log \sigma}{\log_\sigma \ts}
		+ \frac{\ts \log \log \ts}{\log_\sigma \ts} \right)$.
	\end{itemize}	

	It supports 
	\procfont{firstchild(u)}, \procfont{parent(u)},
	\procfont{nextsibling(u)}, \procfont{lca(u,v)},
	\procfont{childrank($u$)}, \procfont{child($u$,$i$)} and \procfont{depth($u$)}
	operations in $\Ocomp(1)$ time;
	moreover
	with additional $\Ocomp(\ts (\log \log \ts)^2 / \log_\sigma \ts)$ bits 
	it can support \procfont{level\_ancestor($v$, $i$)} query.
\end{theorem}

\begin{proof}[Proof of Theorem~\ref{thm:theorem_tree_structure} (and~\ref{thm:theorem_tree_structure_full})]
	We start by proving the part for operations
	\procfont{firstchild(u)}, \procfont{parent(u)},
	\procfont{nextsibling(u)}, \procfont{lca(u,v)},
	as the rest requires additional structures.
	
	First we consider the case for $\sigma \leq \ts^{1-\alpha}$.
	
	To bound the memory consumption we sum
	needed space for \Trefall.
	\Tref{1}, \Tref{3}, \Tref{4} take at most
	$\Ocomp(|\mathcal{T'}|) = \Ocomp(\frac{\ts}{\log_\sigma \ts})$ bits.
	We bound space for \Tref{2} by Theorem~\ref{lem:main_estimation},
	this summand dominates others.
	
	One of the crucial part is that we can perform \procfont{preorder-rank}
	and \procfont{preorder-select} in constant time,
	this allows us to retrieve node labels of tree from $C(\mathcal{T})$ from preorder sequence
	(and thus the cluster) given node of $\mathcal{T'}$ in constant time.
	
	We now give the description of operations, let $u$ denote the node and $u'$ the name of its cluster.
	If the answer can be calculated using only the cluster of $u$
	(i.e.\ when $u$ and the answer is in the same cluster)
	we return the answer using precomputed tables.
	Thus in the following we give the description when the answer cannot be computed within
	the cluster $u'$ alone.

	\smallskip
	\noindent 
	\procfont{firstchild($u$)}:
	Using the structure for degree sequence~\Tref{3} we find
	index $i$ of child of $u'$ which represents cluster containing first child
	of $u$.
	We call \procfont{child($u$, $i$)} on the structure for unlabeled tree $\mathcal{T'}$ \Tref{1},
	to get this cluster, the answer is the root of the first tree in this cluster.
	
	\smallskip
	\noindent 
	\procfont{parent($u$)}
	We call \procfont{childrank($u'$)}
	on structure for $\mathcal{T'}$.
	This gives us index $i$ such that $u'$ is $i$-th node of node $v'$ representing the cluster containing \procfont{parent($u$)}.
	Now we query degree sequence structure,
	as it supports also reverse queries (see Lemma~\ref{lem:label_seq}) obtaining index of port node.
	Finally, we use precomputed table (i.e.\ we query the table which for given index of port node
	and given cluster returns this port node).
	
	\smallskip
	\noindent 
	\procfont{nextsibling}$(u)$:
	We call \procfont{nextsibling($u'$)}
	on structure for $\mathcal{T'}$ \Tref{1} and take root of the first tree
	in the cluster and verify that it has the same parent as $u$.
	
	
	\smallskip
	\noindent  \procfont{lca($u$, $v$)}: let $v'$ be the cluster of $v$.
	We use the structure for $\mathcal{T'}$~\Tref{1}:
	the answer is in the cluster which is represented 
	by node $l =  \procfont{lca($u'$ ,$v'$)}$ of $\mathcal{T'}$
	but we still need to determine the actual node inside the cluster.
	To this end we find nodes $u'', v''$ of $\mathcal{T'}$ such that:
	both are children of $l$,
	they are ancestors of $u'$ and $v'$ respectively,
	and $u''$ and $v''$ connect to some (port) nodes
	$x, y$ such that $x, y$ are in the cluster represented by node $l$
	and \procfont{lca($x, y$)} = \procfont{lca($u, v$)}.
	They can be comptued as follows:
	$u''$=
	\procfont{level\_ancestor($u'$, depth(lca($u'$ ,$v'$))-depth($u'$)-1)},
	the case for $v''$ is analogous.
	Having $u''$ and $v''$ we can, as in the case for \procfont{parent($v$)},
	reverse query the structure for degree sequence~\Tref{3}
	obtaining indices of $x$ and $y$.
	Finally we use precomputed table
	(as we want to find lowest common a
	for port nodes with indices $x$ and $y$ in given cluster).
	
	Note that most tree structures
	allow to find $u'', v''$ without
	calling \procfont{depth} and \procfont{level\_ancestor},
	as they are rank/select structures on balanced parenthesis,
	and it is easy to express this operation using such
	structures~\cite{fullyFunctionalTrees}.
\\\\
For general $\sigma$ we need to only slightly modify our solution.

To encode labels the string $P$ we use results from~\cite{FerraginaV07SimpStat},
which achieves $|P|H_0(P) + |P|\log \log |P|$ bits.
As $|P| = \Ocomp(\ts  / \log_\sigma \ts)$ this gives required bound.

Additionally if $\sigma = \Omega(\ts)$
we do not have to use precomputed tables,
as every cluster have constant number of nodes in it,
thus we can perform operations inside the clusters in constant time.

In other case we use tables of size $\Ocomp(\ts)$, this is still
within bounds, as this is dominated by $\Ocomp(\ts \log \log \ts / \log_\sigma \ts)$.
\end{proof}

Note that in even in Theorem~\ref{thm:theorem_tree_structure} for the case when $\sigma = \ts^{1-\alpha}$,
the guanrantee on the redundancy is $\Ocomp(n)$,
which is of the same magnitude as the size
of the encoding of the tree using parentheses,
so for large alphabets this dominates tree entropy.

Also in Theorem~\ref{thm:theorem_tree_structure}  for 
the case for arbitrary $\sigma$ we get a slightly worse redundancy,
i.e.\ we had $\Ocomp(\ts \log \log \ts / \log_\sigma \ts)$ factor
instead of $\Ocomp(\ts \log \log / \log_\sigma \ts)$.
Still, even in the case of $\sigma = \omega(n^{1-\alpha})$
we can get better bounds (more precisely: $\Ocomp(\ts \log \log_\sigma \ts / \log_\sigma \ts)$,
or $\Ocomp(n)$ for large alphabets)
by encoding string of labels $P$ using structure like~\cite{belazzougui2015optimal};
but at the cost that operations are slower than $\Ocomp(1)$.
%
\\\\
We now prove the rest of Theorem~\ref{thm:theorem_tree_structure},
that is we can add even more operations,
for the rest of the operations we will need more involved
data structure.
\paragraph{Succinct partial sums}
To realize more complex operations
we make use of structure for succinct partial sums.
This problem was widely researched, also in dynamic setting~\cite{succintPartialSums}.
For our applications, however, it is enough to use a basic, static structure
by Raman et al.~\cite{succintDictionariesWithTrees}.
\begin{lemma}\label{lem:partial_sums}
	For a table $|T|=n_t$ of nonnegative integers
	such that $\sum_{i} T[i] \leq n$
	we can construct a structure which answers the following queries in constant time:
	\procfont{sum($i,j$)}: $\sum_{y=i}^{j} T[y]$;
	\procfont{find($x$)}: find first $i$ such that $\sum_{y=1}^{y=i} T[y] \geq x$;
	and consumes $\Ocomp(n_t\log \frac{n}{n_t})$ bits.
\end{lemma}
\begin{proof}[Proof of Lemma~\ref{lem:partial_sums}]
	We exploit the fact that all the numbers sum up to $n$.
	This allows us to store $T$ unary, i.e.\ as string $0^{T[i]}1\ldots 0^{T[n_t]}1$.
	Now using rank/select we can realize desired operations, see~\cite{succintDictionariesWithTrees}
	for details.
	For a string with $n$ zeros and $n_t$ ones
	this structure takes $\log \binom{n+n_t}{n_t} + o(n_t)$ bits.
	This can be estimated as: $\log \binom{n+n_t}{n_t} + o(n_t) \leq n_t \log (\mathrm{e}(n+n_t)/n_t) + o(n)
	= \Ocomp(n_t \log n/n_t)$, as claimed. 
\end{proof}

\paragraph{\procfont{childrank($v$)}, \procfont{child($v, i$)}}
Observe first that if $v$ and its parent
are in the same cluster then \procfont{childrank}$(v)$
can be answered in constant time,
as we preprocess all clusters.
The same applies to $v$ and its children in case of
\procfont{child}$(v,i)$.
Thus in the following we consider only the case when $v$
is a root of a tree in a cluster (for \procfont{childrank}$(v)$)
or it is a port (for \procfont{child}$(v,i)$).

The problem with those operations is
that one port node $p$ connects to multiple clusters,
and each cluster can have multiple trees.
We solve it by storing for each port node $p$
a sequence $T_p = t_{p, 1}, \ldots, t_{p, j}$
where $t_{p, j}$ is the number of children of $p$ in the $j$-th (in left to right order)
cluster connecting to $p$.

Observe that all sequences $T_p$ contain in total $|\mathcal{T'}|-1$
numbers, as each number corresponds to one cluster.
To make a structure we first concatenate 
all sequences according to preorder of nodes
in $\mathcal{T'}$, and if multiple nodes
are in some cluster we break the ties by
left-to-right order on port nodes.
Call the concatenated sequence $T$.
Using structure from Lemma~\ref{lem:label_seq}
for port node $p$ we can find indices $i_1, i_2$
which mark where the subsequence corresponding to $T_p$ starts and ends in $T$,
i.e.\ $T[i_1 \twodots i_2-1] = T_p$.


We build the structure from Lemma~\ref{lem:partial_sums}
for $T$, this takes $\Ocomp(\ts \log \log_\sigma \ts / \log_\sigma \ts)$ bits,
as $|T| = \Ocomp\left(\frac{\ts}{\log_\sigma \ts}\right)$ and all elements in $T$ sum up to at most $\ts$.
We realize \procfont{childrank($v$)} as follows:
let $v'\in \mathcal{T'}$ be a node representing cluster containing $v$ (by the assumption: as a root).
First we find indices $i_1, i_2$ corresponding to subsequence $T_p$,
where $p$ is port node which connects to $v'$.
Let $j$ = \procfont{childrank($v'$)} in $\mathcal{T'}$.
Now it is enough to get \procfont{sum($i_1$, $j-1$)},
as this corresponds to number of children in first $j-i_1$
clusters connected to $p$,
and add to the result the rank of $v$ in its cluster,
the last part is done using look-up tables.

The \procfont{child($v$, $i$)} is analogous:
we find indices $i_1, i_2$ corresponding to $T_p$.
Then we call \procfont{find($i$ + sum($1$, $i_1$-1))}
to get cluster containing \procfont{child($v$, $i$)},
as we are interested in first index $j$ such that
$T[i_1] + \ldots + T[j] > i$.
This way we reduced the problem to find $i'$-th
node in given cluster, this can be done using precomputed tables.
\paragraph{\procfont{depth($v$)}}
The downside of clustering procedure is that
we lose information on depth of vertices.
To fix this, we assign to each edge a non-negative natural \emph{weight} in the following way:
Let $v \in \mathcal{T'}$ be any vertex and
$p$ be a port node in cluster represented
by \procfont{parent($v$)}.
For an edge ($v$, \procfont{parent($v$)})
we assign depth of $p$ in cluster represented by \procfont{parent($v$)}.
For example in Figure~\ref{fig:trees}
for edge $(D,B)$ we assign $1$,
and for edge $(I, B)$ we assign $2$.
In this way the depth of the cluster $C$ (alternatively: depth of roots of trees in $C$)
in $\mathcal{T}$
is the sum of weights of edges from root to node representing $C$.

A data structure for calculating depths is built using
a structure for partial sums:
Consider balanced parentheses representation of $\mathcal{T'}$.
Then we assign each opening parenthesis corresponding to node $v$
weight $w(v,parent(v))$ and each closing 
parenthesis weight $-w(v,parent(v))$.
This creates the sequence of numbers,
for example, for tree from Figure~\ref{fig:trees}
we have:

\noindent
\begingroup
\setlength{\tabcolsep}{5pt}
\begin{tabular}{cccccccccccccccccccccccc}
	( & (& (& )& (& )& (& )& (& )& (& )& (& )& (& )& )& (& (& )& (& )& )& ) \\
	0 & 1 & 1 & -1 & 1 & -1 & 1 & -1 & 2 & -2 & 2 & -2 & 2 & -2 & 2 & -2
	& -1 & 1 & 1 & -1 & 1 & -1  & -1 & 0
\end{tabular}
\endgroup

\noindent
Then we can calculate depth of a cluster by calculating the prefix sum.
Observe that our partial sums structure does not work on negative numbers,
but we can solve that by creating two structures,
one for positive and one for negative number
and subtract the result.
Finally we use look-up table to find the depth in the cluster

The total memory consumption is bounded
by $\Ocomp(\ts \log \log_\sigma \ts / \log_\sigma \ts)$,
as all weights sum to at most $\ts$.
\paragraph{\procfont{level\_ancestor($v$, $i$)}}\label{sec:lvl_ancestor}
We assign weights to edges as in the case of \procfont{depth} operation.
This reduces the level ancestor in $\mathcal T$ to
\emph{weighted level ancestor} in $\mathcal T'$;
in this problem
we ask for such ancestor $w$ of $v$ that sum of
weights on the path from $w$ to $v$ is at least $i$
and $w$ is closest node to $v$ in the terms of number of nodes on the path
(note that there may not exist a node for which the sum is equal exactly to $i$).
The redundancy obtained for level\_ancestor operation is slightly
worse than for previous operations,
but not worse than most of the other structures~\cite{ultraSuccintTrees, geary2006succinct}
supporting this operation.
Observe that each edge has weight of order $\Ocomp(\log_\sigma \ts)$.
From the following theorem we get that additional
$\Ocomp(\ts (\log \log \ts)^2 / \log_\sigma \ts)$ bits is sufficient.
\begin{lemma}\label{lem:structure_weighted_level_ancestor}
	Let $\mathcal{T'}, |\mathcal{T'}| = t$ be a tree
	where each edge is assigned a weight of at most $\Ocomp(\log n)$,
	for some $n$.
	We can builds structure which consumes
	$\Ocomp(t (\log \log n)^2)$ bits of memory
	and allows to answer weighted level ancestor queries
	in $\Ocomp(1)$ time.
\end{lemma}

With the structure from Lemma~\ref{lem:structure_weighted_level_ancestor} the query 
\procfont{level\_ancestor($v$, $i$)}
is easy:
first we check if the answer is in the same cluster using
preprocessed array.
If not we find cluster containing the answer,
we do it by asking for \procfont{level\_ancestor($C$, $i$-depth($C$, $v$))},
where \procfont{depth($C$, $v$)} is depth of $v$ in cluster $C$ containing $v$.
There is similar problem as in the case of \procfont{lca} query,
that is we also need to find the port node
on path from given vertex $v$ to its $i$-th ancestor.
This may be solved in the same manner as in the case for \procfont{lca}.

Now we give the construction for weighted level ancestor structure.
Note that there are multiple ways of doing this~\cite{geary2006succinct,fullyFunctionalTrees,
ultraSuccintTrees}.
We use the tree partitioning approach,
yet the one that operates on sequence of numbers, as in case for depth,
should also be applicable, \cite{ultraSuccintTrees} shows similar method
(and uses same additional space),
yet for simplicity we choose to stick with solution which partition
the tree into subtrees as we already defined most of the required machinery.
Note that tree partitioning method~\cite{geary2006succinct},
which we refer to, partitioned the tree a few times,
we do it once and use stronger result~\cite{fullyFunctionalTrees} for the 
simplicity of proof.
Also, the partitioning from~\cite{geary2006succinct}
may be used instead of our method.
\begin{proof}[Proof of Lemma~\ref{lem:structure_weighted_level_ancestor}]
	We use the idea from~\cite{geary2006succinct}.
	We first cluster the tree according to Lemma~\ref{lem:clustering_procedure}
	with $m = \Theta(\log^3 n)$.
	We obtain a smaller tree $\mathcal{T''}$ of size
	$|\mathcal{T''}| = \Ocomp\left(\frac{t}{\log^3 n}\right)$.
	We store labels of $\mathcal T''$ and the descriptions of clusters naively, without the entropy coder.
	We also store additional structure for navigation of $\mathcal{T''}$,
	including degree sequences,
	observe that it takes at most $\Ocomp(t)$ bits.
	
	As before (i.e.\ for the case of depth), we assign weights to edges, but we have to remember 
	that input tree is weighted as well.
	Let $v \in \mathcal{T'}$ be any vertex and
	$p$ be a port node in cluster represented \procfont{parent($v$)}.
	For an edge $(v, \procfont{parent(v)})$, where $v \in\mathcal{T''}'$,
	we assign sum of weights on the path from $p$ to root of the cluster containing $p$.
	
	Observe that the weights in the $\mathcal{T''}$
	are of order $\Ocomp(\log^4 n)$.
	For such weighted tree we can build structure which supports
	\procfont{level\_ancestor} queries in $\Ocomp(1)$ time
	and use $\Ocomp(|\mathcal{T''}| \log^2 n) = o(n)$ bits,
	using result from~\cite{bender2004level}.
	
	Now for each smaller tree we can build
	structure from~\cite{fullyFunctionalTrees}.
	For a tree of size $t'$ and with weights limited 
	by $\Ocomp(\log^4 n)$ this structure takes $\Ocomp(t' \log t' \log (t'\log^4 n))
	= \Ocomp(t'(\log \log n)^2)$ bits.
	Summing over all trees, we get $\Ocomp(t (\log \log n)^2)$, as claimed.
	For each such tree we additional store the information
	which nodes are port nodes in a way that
	allow to retrieve $i$-th port node and.
	This can be achieved by storing the bitmap
	for each small tree (in which each $j$-th element indicates if $j$ leaf is port node or not)
	and applying rank/select structure (this consumes $\Ocomp(t)$ bits for all trees).
	Observe that we can even explicitly list all of the port nodes:
	we do not have to use space efficient solution,
	as there are at most $|\mathcal{T''}|$ such nodes,
	so even consuming $\Ocomp(\log^2 n)$ bits for port node
	is sufficient.
	
	We perform \procfont{level\_ancestor} operation in the same
	manner as previously described when applying Lemma~\ref{lem:structure_weighted_level_ancestor},
	i.e.\ we first check if the answer is in the same tree,
	if yes we can output the answer as 
	we can perform operations on small trees in $\Ocomp(1)$ time
	~\cite{fullyFunctionalTrees},
	if not we use combination of \procfont{depth} and $\procfont{level\_ancestor}$
	queries on structure for $\mathcal{T''}$
	in the same way as we did for \procfont{lca} query
	(see proof of Theorem~\ref{thm:theorem_tree_structure}).
	It is possible as structure~\cite{bender2004level} supports all of the required
	operations (or can be easily adapted to support
	by adding additional tree structure, 
	as the structure for $\mathcal{T''}$ can consume up to  $\Ocomp(\log^2 n)$ bits per node,
	in particular this means that we can even preprocess all answer for
	$\procfont{depth}$ queries for $\mathcal{T''}$.
	
	The only nontrivial thing left is that we would like to
	not only find a node in our structure but also
	find corresponding node in structure for $\mathcal{T'}$
	which supported other operations.
	To this end we show that we can return preorder position
	of given node, this is sufficient as $\mathcal{T'}$
	has \procfont{preorder-select} operation.
	
	To this end we explicitly store preorder and subtree size
	for each port node, observe that this consumes
	at most $\Ocomp(\mathcal{T''} \log t) = \Ocomp(t)$ bits.
	Now given a node $v$ in some cluster
	to find preorder rank of this node in $\mathcal{T'}$ we first
	find preorder rank of $v$ in cluster $C$ containing $v$
	then sum it with the rank of port node $p_v$ which is connected to $C$
	and the sizes of subtrees $T_i$ which are connected to port nodes $c_i$ of $C$,
	such that each $c_i$ precede $v$ in preorder ordering in $C$.
	To find sum of sizes of $c_i$ we first find how many $c_i$'s precede $v$
	in preorder ordering in $C$, to this end we   use rank/select structure for binary vector $B_C$
	where $B_C[i] = 1$ if and only if $i$-th node according to preorder ordering in $C$.
	Then we use the structure for partial sums
	(again, we do not have to use succinct structure
	as there are at most $\Ocomp(\mathcal{T''})$ elements in total).
	
\end{proof}

\section{Additional material for Section~\ref{sec:even_succinter}}
We start by stating that by that we can get
better estimation by encoding each class separately.
\begin{lemma}\label{lemma:main_estimation_boosting}
	Let $\mathcal{T'}$ be a labeled cluster tree
	from Lemma~\ref{lem:cluster_structure}
	for parameter $m$,
	obtained from $\mathcal{T}$.
	For each $k$-letter context $K_i$ let $P_{K_i}$
	be a concatenation of labels of $\mathcal{T'}$
	which are preceded by this context (i.e.\ each root $v$ in each cluster
	is preceded by the context $K_i$ in $\mathcal{T}$).
	Then \emph{all} of the following inequalities hold:
	
	\begin{enumerate}\setlength\itemsep{0.1em}
		\item \label{lemma_boosting:case1}
		$ \sum_i |P_{K_i}| H_0(P_{K_i}) \leq \ts H(\mathcal{T}) + \ts H_k(L) +
		\Ocomp\left(\frac{\ts \log m}{m} \right)$;
		\item \label{lemma_boosting:case2} 
		$ \sum_i |P_{K_i}| H_0(P_{K_i}) \leq 
		\ts H_k(\mathcal{T}|L) + \ts H_k(L)+
		\Ocomp\left(\frac{\ts \log m}{m} \right)$;
		\item \label{lemma_boosting:case3}
		$\sum_i |P_{K_i}| H_0(P_{K_i}) \leq
		\ts H(\mathcal{T}) + \ts H_k(L|\mathcal{T})+
		\Ocomp\left(\frac{\ts \log \sigma}{m}
		+ \frac{\ts \log m}{m} \right)$.
	\end{enumerate}
\end{lemma}

\begin{proof}[Proof of Lemma~\ref{lemma:main_estimation_boosting}]
	The proof is very similar to the proof of Theorem~\ref{lem:main_estimation}.
	For each $T_{K_i}$ we apply the Lemma~\ref{theoremP}:
	we use almost the same values of $q$ function for each cluster
	but we do not need to multiply it by $q(K_C)$,
	i.e.\ we define $q(C)= q(N_C)\cdot q(V_C)$.
	For detailed definition of $q$ see proof of Theorem~\ref{lem:main_estimation}.
	It is easy to check that without this factor we arrive at the claim.
\end{proof}

\begin{lemma}
	\label{lem:huffman_encode_boosting}
	Let $\mathcal{T'}$ be a labeled cluster tree
	from Lemma~\ref{lem:cluster_structure}
	for parameter $m$,
	obtained from $\mathcal{T}$.
	Let $P$ be a string obtained by concatenation of labels of
	$\mathcal{T'}$.
	Then we can encode $P$ in a way, that, given context $K_{P[i]}$,
	we can retrieve each $P[i]$ in constant time.
	The encoding is bounded by \emph{all} of the following values:
	
	\begin{enumerate}\setlength\itemsep{0.1em}
		\item
		$\ts H(\mathcal{T}) + \ts H_k(L) +
		\Ocomp\left(\frac{\ts (\log m + \log \log \ts)}{m}  +
		\sigma^{k+m} \cdot 2^m \cdot \log \ts \right)$;
		\item
		$\ts H_k(\mathcal{T}|L) + nH_k(L)+
		\Ocomp\left(\frac{\ts(\log m + \log \log \ts)}{m} +
		\sigma^{k+m} \cdot 2^m \cdot \log \ts \right)$;
		\item
		$\ts H(\mathcal{T}) + \ts H_k(L|\mathcal{T})+
		\Ocomp\left(\frac{\ts \log \sigma}{m}
		+ \frac{\ts (\log m + \log \log \ts)}{m} + 
		\sigma^{k+m} \cdot 2^m \cdot \log \ts \right)$.
	\end{enumerate}
\end{lemma}

\begin{proof}[Proof of Lemma~\ref{lem:huffman_encode_boosting}]
	We use Lemma~\ref{lemma:main_estimation_boosting}.
	For each $P_{K_i}$ we generate codes using Huffman encoding,
	this allows us to encode each $P_{K_i}$ using
	$|P_{K_i}|H_0(P_{K_i}) + |P_{K_i}|$ bits,
	as we lose at most $1$ bit per code (see~\cite{FerraginaV07SimpStat} for example),
	plus additional $\Ocomp(2^m \sigma^m \log |\mathcal{T'}|) \leq \Ocomp(2^m \sigma^m \log \ts)$ bits for Huffman dictionary.
	Summing this over all contexts $K_i$ yields the bound, by Lemma~\ref{lemma:main_estimation_boosting}.
	
	Denote $c_v$ as Huffman code for vertex $v$ of $\mathcal{T'}$.
	Let $T_C$ be the concatenation of all codes $c_v$ according to order in string $P$.
	As each code is of length at most $\Ocomp(\log \tps)$,
	given its start and end in $T_C$
	we can decode it in constant time.
	We store the bitmap of length $|T_C|$ where $T_C[j]=1$ if and only if
	at $j$-th is the beginning of some code.
	This bitmap has length at most $\tps \log \ts$ and has $\tps$ ones.
	For such a bitmap we build rank/select structure,
	using result by Raman et al.~\cite{succintDictionariesWithTrees}
	this takes $\Ocomp(\tps \log \log \ts)$ bits.
	Using rank/select we can retrieve the starting position of $i$-th
	code by simply calling select($i$) (same idea was used in~\cite{GonzalezNStatistical}).
\end{proof}


Now we would like to simply apply Lemma~\ref{lem:huffman_encode_boosting},
which says that we can encode labels of $\mathcal{T'}$ more efficiently,
yet there is one major difficulty:
the Lemma states that to decode $P[i]$ we need to know the context.
The idea is that we choose $|\mathcal{T'}|/d$ nodes for
which we store the context, for rest we can retrieve the contexts
in time $\Ocomp(d)$ by traversing $\mathcal{T'}$ and decoding
them on the way.

\begin{proof}
	[Proof of
	Theorem~\ref{thm:theorem_tree_structure_boosting}]
	We use almost the same structures as in the simpler case,
	i.e.\ the structure from Theorem~\ref{thm:theorem_tree_structure},
	the only difference is that instead encoding preorder
	sequence $P$ with structure using space proportional to zeroth
	order entropy we apply Lemma~\ref{lem:huffman_encode_boosting}.
	We choose $m = \beta \log_\sigma \ts$ so that $2\beta + \alpha < 1$
	and $\beta < \frac{1}{8}$, so that the precomputed tables use $o(\ts)$ space.
	As each operation in Theorem~\ref{thm:theorem_tree_structure}
	accessed elements of $P$ constant number of times
	it is sufficient to show how to access it in time $\Ocomp(\log \ts / \log \log \ts)$ time.
	By Lemma~\ref{lem:huffman_encode_boosting} this leaves us with problem of
	finding context for each node in aforementioned complexity.
	
	
	Let $d = \lceil \log \ts / \log \log \ts \rceil$.
	We choose at most
	$\Ocomp(\tps \log \log \ts / \log \ts)
	$=$\Ocomp(\tps/d)$
	nodes, for which we store the context explicitly,
	in the following way:
	We store the context for the root using
	$\lceil k \log \sigma\rceil$ bits.
	We partition the nodes into $d$ classes $C_i$ depending on their depth modulo $d$,
	i.e.\ in the class $C_i$ there are nodes at depth $dj + i, j \geq 0$.
	Then there is a class having at most $\mathcal{T'}/d$ such nodes,
	we choose all nodes in this class and store their contexts.
	
	Now we show, assuming we know the contexts for chosen nodes, 
	that given a node we can retrieve its context in $\Ocomp(d)$ time.
	If we want to compute the context for some node $v$
	we first check whether its stored explicitly.
	If not, we look at nodes on path from $v$ to the root,
	until we find a node $u$ which has its context stored.
	Call the visited nodes $v, v_1, v_2, \ldots, v_i, u$.
	Observe that we visited at most  $\Ocomp(d)$ nodes that way.
	As we know the context for $u$, now we can decode the node $u$,
	and determine the context for $v_i, v_{i-1},\ldots, v_1, v$.
	To read the labels in $u$ which precede $v_i$ in constant time
	we first find port node which connects to $u$ (as in the proof of Theorem~\ref{thm:theorem_tree_structure})
	and use the precomputed tables.
	
	The only nontrivial thing left to explain is how to store 
	the contexts for chosen $\Ocomp(\tps/d)$ nodes
	and check which nodes have their contexts stored.
	We concatenate all contexts for chosen $\Ocomp(\tps/d)$ nodes
	according to their order in preorder ordering.
	On top of that we store binary vector $B$ which satisfies
	$B[i]=1$ if and only if $i$-th node of $\mathcal{T'}$
	in preorder ordering has its context stored.
	We build rank/select structure for $B$,
	as we have \procfont{preorder-rank} and \procfont{preorder-select}
	operation for $\mathcal{T'}$ in constant time, for a given node we
	can check in constant time if the node have its index stored or not.
	As each context has the same bit-length to decode context for node
	which is $j$-th in preorder ordering we look at position $(j-1)\lceil k \log \sigma\rceil$.
	
	The total space for storing the context is $\Ocomp(\tps k \log \sigma / d)
	= \Ocomp(\tps \log \log \ts)$,
	summing that up with space bound from Lemma~\ref{lem:huffman_encode_boosting}
	yields the claim.
\end{proof}

\section{Additional material for Section~\ref{sec:label_operations}}
\begin{proof}[Proof of Theorem~\ref{thm:label_operations}]
The first part of the theorem is easy: if $\sigma$ is constant we can
construct a separate structure for each letter.
For each letter $a$ we build a separate degree sequence, level\_ancestor structure
and depth structure; observe that all of those structures support
the weighted case when we assign each vertex weight of $0$ or $1$,
so it is sufficient to assign nodes labeled with $a$ value $1$
and for the rest value $0$. 
Similar idea was mentioned in~\cite{tsur2015succinct,he2014framework,geary2006succinct}.

For the next two parts we show how to adapt rank/select structures over large alphabets
to support childrank/childselect queries.

For the second part, when $\sigma = \Ocomp(\log^{1+o(1)} \ts )$ we
use result by Belazzougui et al.~\cite{OptimalLowerUpper}.
They show (at discussion at above Theorem 5.7 in \cite{OptimalLowerUpper})
how for the sequence $S$ divided into $\Ocomp(|S|/m)$ blocks of length at most $m$, for some $m$,
construct rank/select structure for large alphabets, assuming that we can
answer queries in time $\Ocomp(t)$ in blocks,
such that it takes $\Ocomp(t)$ time for query and consumes additional
$|S| \log \frac{\sigma}{m} + \Ocomp\left(|S| + \frac{(\sigma |S| / s)
\log \log (\sigma |S| / s)}{\log \sigma |S| /s}\right)$ bits
(in~\cite{OptimalLowerUpper} the assumption is that we can answer queries in blocks in $\Ocomp(1)$
time but as the operations on additional structure cost constant time,
our claim also holds).
The solution uses only succinct bitmaps by Raman et al.~\cite{raman2003succinct} and precomputed tables
for additional data.
Now we can define sequence $S$ as concatenation of labels of roots of cluster,
where clusters are ordered by preorder ordering,
this gives our blocked sequence (where blocks correspond to clusters).
Observe that labeled childrank/childselect operations can easily be reduced
to labeled rank/select in string $S$, all we need to do is to know where
the sequence for children of a given vertex $v$ begins in $S$.
Fortunately, this can be done in same manner as in Lemma~\ref{lem:label_seq},
that is, we use structure for degree sequence.
As in our case $m = \log_\sigma \ts$, we can store precomputed tables to answer rank/select
queries for each cluster.
The additional space is $o(|S| \log \sigma)$ and clearly $|S| \leq \ts$.

For the last part we use Lemma 3 from~\cite{BarbayRankSelect}.
The lemma states that for a string $|S|$ if, for a given $i$,
we can access $i$-th element in time $\Ocomp(t)$
then we can, using additional $o(|S|\log\sigma)$ bits,
support labeled rank/select operations
in time $\Ocomp(t \log \log^{1+\epsilon} \sigma)$.
We use the same reduction as in the case for $\sigma = \Ocomp(\log^{1+o(1)} |S|)$,
i.e.\ we set $S$ as concatenation of labels of roots of clusters, where clusters
are ordered by preorder ordering.
As in previous case, we use node degree sequence
and tree structure to~retrieve $i$-th character in $|S|$
(i.e.\ we first find appropriate cluster and use precomputed tables).
\end{proof}

\bibliographystyle{ACM-Reference-Format}
\bibliography{references}

\end{document}